\documentclass[11pt]{article}
\usepackage{authblk}
\usepackage[small,compact]{titlesec}
\usepackage{amsmath}
\usepackage{amssymb}
\usepackage{amsthm}
\usepackage[latin1]{inputenc}
\usepackage{graphicx}
\usepackage{color}
\usepackage{mathptmx}
\usepackage{pgfplots}
\usepackage{hyperref}
\usepackage{a4wide} 
\usetikzlibrary{plotmarks}

\newcommand{\Mac}[2]{\mathcal{M}_{#1,#2}}

\def\LT{\operatorname{LT}}
\def\HF{\operatorname{HF}}

\newtheorem{theorem}{Theorem}
\newtheorem{proposition}[theorem]{Proposition}
\newtheorem{lemma}[theorem]{Lemma}
\newtheorem{corollary}[theorem]{Corollary}
\newtheorem{example}{Example}

\newtheorem{definition}{Definition}

\newcounter{algo}
\newcommand\ligne[2][ ]{\refstepcounter{algo}{\tiny \arabic{algo}. }#1 \hspace*{#2mm}}

\begin{document}
\title{On the Complexity of the $F_5$ Gröbner basis Algorithm}

\author[1]{Magali Bardet} 
\author[2,3,4]{Jean-Charles Faug\`ere}
\author[5,6,7]{Bruno Salvy}
\affil[1]{\'Equipe C\&A, LITIS, Universit{\'e} de Rouen}
\affil[2]{UPMC, Univ Paris 06, LIP6}
\affil[3]{CNRS, UMR 7606, LIP6}
\affil[4]{PolSys Project, Inria}
\affil[5]{AriC Project, Inria}
\affil[6]{LIP, ENS de Lyon}
\affil[7]{CNRS, UCBL, Université de Lyon.}

\maketitle
\begin{abstract}
  We study the complexity of Gr\"obner bases computation, in
  particular in the generic situation where the variables are in
  simultaneous Noether position with respect to the system.

We give a bound on the number of polynomials of degree $d$ in a
Gröbner basis computed by Faugère's $F_5$
algorithm~(\cite{Faugere02}) in this generic case for the
grevlex ordering (which is also a bound on the number of polynomials
for a reduced Gröbner basis, independently of the algorithm used). 
Next, we analyse more precisely the structure of the polynomials in 
the Gr\"obner bases with signatures that $F_5$ computes 
and use it to bound the complexity of the algorithm.

Our estimates show that the version of~$F_5$ we analyse, which uses only standard Gaussian elimination techniques, outperforms row reduction of the Macaulay matrix with the best known algorithms for moderate degrees, and even for degrees up to the thousands if Strassen's multiplication is used. The degree being fixed, the factor of improvement grows exponentially with the number of variables.
\end{abstract}


%%%%%%%%%%%%%%%%%%%%%%%%%%%%%%%%%%%%%%%%%%%%%%%%%%%%%%%%%%%
\section*{Introduction}
The complexity of Gr\"obner bases has been the object of extensive
studies. It is well-known that in the worst-case, the complexity is
doubly exponential in the number of variables. This is the result of a
series of works both on lower bounds by~\cite{MaMe82,Huynh86} and on
upper bounds, first in characteristic~0 by~\cite{Giusti84,MoMo84} and
then in positive characteristic by~\cite{Dube90}.

These worst-case estimates have led to the unfortunately widespread
belief that Gr\"obner bases are not a useful tool beyond toy
examples. However, it has been observed for a long time that the
actual behaviour of Gr\"obner bases implementations can be quite
efficient. For instance, the matrix-$F_5$ algorithm that we analyse
in this article, itself a downgraded version of
Faugère's~$F_5$ algorithm (\cite{Faugere02}) and a
particular case of~\cite{FaugereRahmany2009}, has
given surprisingly good results on a cryptographic challenge (see
\cite{FaJo03} where a set of $80$ dense polynomials in $80$ variables
was solved by this algorithm). This motivates an investigation of the
complexity of Gr\"obner basis algorithms for useful special classes of
polynomial systems.

In this article, we concentrate on the important case of homogeneous
systems. Any system can be brought into this form by adding a variable and homogenizing.
It is classical that the computation of Gr\"obner bases can
be performed by linear algebra on a large matrix that has been
described precisely by~\cite{Macaulay02}. The explicit relation with
Gr\"obner bases can be found in the works
of~\cite{Lazard83} and~\cite{Giusti84,Giusti85}. From there, a simple statement
of a complexity bound is the following.
\begin{proposition}\label{propupper} Let $(f_1,\dots,f_m)$ be a system
  of homogeneous polynomials in $k[x_1,\dots,x_n]$ with $k$ an
  arbitrary field. The number of operations in $k$ required to compute
  a Gr\"obner basis of the ideal~$\mathcal{I}$ generated
  by~$(f_1,\dots,f_m)$ for a graded monomial ordering {\em up to
    degree $D$} is bounded by
\[O\left(mD\,\binom{n+D-1}{D}^{\!\!\omega}\right), \mbox{ as } D \to \infty\]
where $\omega$ is the exponent of matrix multiplication over~$k$.
\end{proposition}
The terminology and notations
relative to Gr\"obner bases are recalled in Section~\ref{sec1}, and we
generally follow~\cite{CoLiOS97}. The simple proof of this proposition
is given in Section~\ref{sec1}. For the notation~$\omega$ and related notions, 
we refer to~\cite{GathenGerhard2003}.

Getting a ``small'' bound on the highest degree of the elements of the
Gr\"obner basis then leads to good complexity estimates. Such a bound
is available for regular systems in the graded-reverse-lexicographical
order (grevlex). In this situation, \cite{Lazard83} has shown that
after a generic linear change of coordinates, a bound is given by the
index of regularity of the ideal, which is itself bounded by
\begin{equation}\label{mac_bound}
\text{Macaulay's\ bound:}\qquad {i_{\text{reg}}} \le\sum_{i=1}^m (d_i-1) + 1,
\end{equation}
where~$d_i=\deg(f_i)$. This bound is named after \cite{Macaulay02},
who obtained it as an upper bound on the degree of intermediate
polynomials used in the computation of a resultant of generic
multivariate polynomials.

Taking~$m=n-\ell$ (with $\ell\ge0$) and injecting Macaulay's bound~\eqref{mac_bound} into the
upper bound of Proposition~\ref{propupper} leads to a general
asymptotic bound for the number of operations:
\begin{equation}\label{intro-upper-bound}
  \left(\frac{\delta^\delta}{(\delta-1)^{\delta-1}}\right)^{\!\!\omega (n-\ell)}n^{2-\omega/2}\left((\delta-1)\left(\frac{\delta}{2\pi(\delta-1)^3}\right)^{\!\!{\omega}/{2}}\!\!+O(1/n)\right),\quad n\rightarrow\infty,
\end{equation}
where $\delta$, assumed to be larger than~1, is the arithmetic mean of the
$d_i$'s. (When $\delta=1$, the system is linear.)

Thus in this case, we have a complexity which is \emph{simply}
exponential in the number of variables. Since in this case, if the
field is algebraically closed, by B\'ezout's bound, the degree of the
variety is also exponential, the result can be interpreted as a
polynomial complexity in some size of the result.  No change of
variable is necessary when the dimension is~0. Otherwise, without a
generic linear change of coordinates, the bound does not hold in
general, as observed by~\cite{MoMo84}.

These results can be made effective by a careful study of the required
genericity condition. Indeed, \cite{LJ84} shows that a sufficient
condition for the bound to hold is that the variables be in
\emph{simultaneous Noether position} with respect to the polynomial
system. (The definition is recalled in Section~\ref{sec1}). If the
system is regular but the variables are not in simultaneous Noether
position, and the field is sufficiently large, then a linear change of
variables can be exhibited that puts the variables in this position.
The complexity of actually finding such a linear change of variables
in the worst case has been studied by~Giusti~(\cite{Giusti88}, \S5.6) and later by~\cite{GiHe93}. It is used as
an ingredient to compute the dimension in small
complexity~(\cite{GiHaLeMaSa00}). The name ``simultaneous Noether
position'' for this situation has been used at least since the work
of~\cite{KrPa96}.

This simply exponential behaviour being established, we are interested
in sharpening the complexity estimates. This is important in order to
compare various algorithms precisely, including approaches to
polynomial system solving that do not use Gr\"obner bases, such as
developed by~\cite{GiLeSa01}.  We concentrate on systems with
variables in simultaneous Noether position. This forms the basis for
many other applications, either by changes of coordinates as we have
just indicated, or by changes of order following~\cite{FaGiLaMo93}, or
by other techniques as developed for instance by
\cite{LaLa91,Lakshman91,HaLa05}.

Most algorithmic variants of
Buchberger's~algorithm (\cite{Buchberger65})  spend
part of their time computing reductions to~0, which is why many
criteria and strategies have been developed over the years. An
assessment of the efficiency of these strategies is obtained for
instance by a comparison of their complexity for $m=n-\ell$ polynomials in $n$ variables with the
bound~\eqref{intro-upper-bound}, for an arbitrary fixed $\ell\ge0$.  We obtain such a complexity estimate
for a specific algorithm, namely Faugère's~$F_5$
algorithm~(\cite{Faugere02}). This algorithm has been the first one to introduce signatures in order to detect efficiently useless reductions to zero.
Since then, many researchers have worked on understanding the new criteria behind $F_5$, which has led to new variants of the signature-based approach.
\cite{preprint3} give a detailed introduction to this topic.
In Section~\ref{sec2}, we present and analyse the matrix-$F_5$ version of the algorithm. A consequence of our results
is the following estimate.
\begin{theorem}\label{thm:nF5}
  Let~$(f_1,\dots,f_m)$ be a system of homogeneous polynomials of
  identical degree~$\delta\ge2$ in $k[x_1,\dots,x_n]$ with $m=n-\ell$ and $\ell\ge 0$, with respect to
  which $(x_1,\dots,x_n)$ are in simultaneous Noether position. Then
  the number of arithmetic operations in~$k$ required \emph{by
    Algorithm matrix-$F_5$} to compute a Gr\"obner basis for the
  grevlex order is bounded by a function of $\delta,\ell,n$ that behaves asymptotically as
  \begin{equation}
    B(\delta)^{n}n\,(A(\delta,\ell)+O(1/n)),\quad
    n\rightarrow\infty,\label{eq:asymptNF5}
  \end{equation}
  when $\ell$ and $\delta$ are $O(1)$. There, the coefficients $B(\delta)$ and $A(\delta,\ell)$ are given by \[B(\delta)=\frac{\left(\frac{\lambda_0+1}{\lambda_0}\right)^{2\delta}-1}{\frac{1}{\lambda_0^2}-\frac{1}{(\lambda_0+1)^2}}\quad \text{and}\quad A(\delta,\ell) = \frac{1-\delta^{-1}}{2
    {\pi }}\cdot
  \frac{\left(1+\lambda_0^{-1}\right)^3-1}{(1+\lambda_0)^{1+\ell}},\]
$\lambda_0$ being the unique
  positive root between $\frac{\delta-1}2$ and $\delta-1$ of
\[\left(\frac{\lambda+1}{\lambda}\right)^{2\delta}=\frac{1}{1-\delta\frac{(\lambda+1)^2-\lambda^2}{(\lambda+1)^3-\lambda^3}}.\]
Moreover, the dominant term~$B(\delta)$ is bounded between $\delta^3$ and~$3\delta^3$.
\end{theorem}

Explicit values of this bound~\eqref{eq:asymptNF5}, called the
$F_5$-bound, are given in
Table~\ref{table:complexite}.
\begin{table}[h]
\centerline{
   \begin{tabular}{|l|l|l|l|l|l|l|l|l|l|l|l|l|l|l|l|l|l|l|l|l|l|}
\hline
$\delta$ & 2& 3& 4& 5& 6& 7& 8\\
\hline
$B(\delta)^n$ & ${2}^{ 4.29\,n}$ & $2^{ 6.16\,n}$ & $2^{ 7.44\,n}$ & $2^{ 8.43\,n}$ & $2^{
  9.23\,n}$ & $2^{ 9.90\,n}$ & $2^{ 10.5\,n}$ \\
\hline
=& $\left( {2}^{n} \right) ^{ 4.3}$ & $\left( {3}^{n} \right) ^{ 3.9}$ & 
 $ \left( {4}^{n} \right) ^{ 3.7}$ & $\left( {5}^{n} \right) ^{ 3.6}$ & 
 $ \left( {6}^{n} \right) ^{ 3.6}$ & $\left( {7}^{n} \right) ^{ 3.5}$ & 
 $ \left( {8}^{n} \right) ^{ 3.5}$ \\
 \hline
=& $(2.5)^n2^{3n}$&$(2.7)^n3^{3n}$&$(2.7)^n4^{3n}$&$(2.8)^n5^{3n}$&$(2.8)^n6^{3n}$&$(2.8)^n7^{3n}$&$(2.8)^n8^{3n}$ \\
\hline
   \end{tabular}}
   \caption[Asymptotic Behaviour of the $F_5$-bound]{Asymptotic Behaviour of the $F_5$-bound (Equation~\eqref{eq:asymptNF5}), in terms of the B\'ezout bound $\delta^n$.\label{table:complexite}}
\end{table}

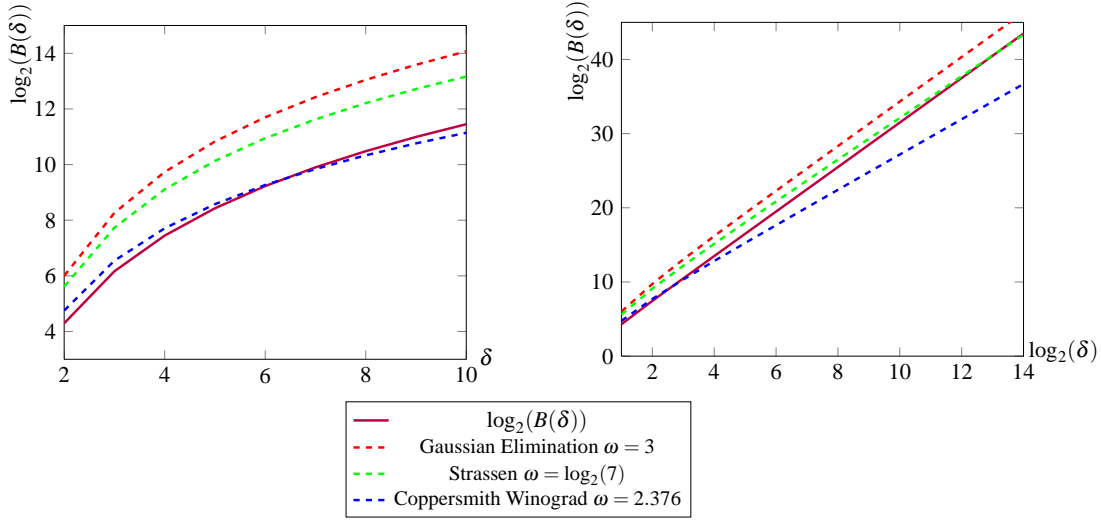
\begin{figure}[h]
\centering
\begin{minipage}[c]{0.47\linewidth}
  \begin{tikzpicture}[scale=0.78]
    \pgfplotsset{every axis legend/.append style={
        at={(1,0.2)},anchor= east}}
    \begin{axis}[ xmin=2, xmax=10, ymin=3, ymax=15, xlabel
      style={at={(1.05,0.05)},anchor=south}, ylabel
      style={at={(0.07,1.09)},anchor=east}, xlabel={\(\delta\)},
      ylabel={\(\log_2(B(\delta))\)}, cycle list={%
        {purple,very thick}, {red,dashed, very thick},
        {green,dashed,very thick}, {blue,dashed,very thick},
        {yellow,dashed,very thick}}, legend style={ at={(0.7,-0.3)},
        anchor=west }]
      \addplot coordinates { (2,4.294889968) (3,6.164453788)
        (4,7.446763612) (5,8.429308942) (6,9.227401400)
        (7,9.899960455) (8,10.48137341) (9,10.99352583)
        (10,11.45123225) };
      \addplot coordinates {(2,6.0000) (3,8.2646) (4,9.7353)
        (5,10.829) (6,11.700) (7,12.425) (8,13.046) (9,13.588)
        (10,14.070)
      };
      \addplot coordinates {(2,5.6145) (3,7.7338) (4,9.1100)
        (5,10.133) (6,10.949) (7,11.627) (8,12.208) (9,12.715)
        (10,13.166) };
      \addplot coordinates {(2,4.7519) (3,6.5456) (4,7.7103)
        (5,8.5765) (6,9.2667) (7,9.8406) (8,10.332) (9,10.762)
        (10,11.143) };
      \legend{$\log_2(B(\delta))$,\small{Gaussian Elimination
          $\omega=3$},\small{Strassen
          $\omega=\log_2(7)$},\small{Coppersmith Winograd $
          \omega=2.376$}
      }
    \end{axis}
  \end{tikzpicture}
\end{minipage}
\begin{minipage}[c]{0.47\linewidth}
  \begin{tikzpicture}[scale=0.78]
    \pgfplotsset{every axis legend/.append style={
        at={(1,0.2)},anchor= east}}
    \begin{axis}[ xmin=1, xmax=14, ymin=0, ymax=45, xlabel
      style={at={(1.1,0.05)},anchor=south}, ylabel
      style={at={(0.07,1.09)},anchor=east},
      xlabel={\(\log_2(\delta)\)}, ylabel={\(\log_2(B(\delta))\)},
      cycle list={%
        {purple,very thick}, {red,dashed, very thick},
        {green,dashed,very thick}, {blue,dashed,very thick},
        {yellow,dashed,very thick}}, legend style={ at={(0.7,-0.3)},
        anchor=west }]
      \addplot coordinates { (1,4.294889968) (2,7.446763612)
        (3,10.48137341) (4,13.48984364) (5,16.49195018)
        (6,19.49247614) (7,22.49260758) (8,25.49264044)
        (9,28.49264866) (10,31.49265071) (11,34.49265121)
        (12,37.49265135) (13,40.49265138) (14,43.49265140) };
      \addplot coordinates { (1,6.0000) (2,9.7353) (3,13.046)
        (4,16.190) (5,19.260) (6,22.294) (7,25.311) (8,28.320)
        (9,31.324) (10,34.326) (11,37.327) (12,40.327) (13,43.328)
        (14,46.328) };
      \addplot coordinates { (1,5.6145) (2,9.1100) (3,12.208)
        (4,15.150) (5,18.022) (6,20.863) (7,23.685) (8,26.501)
        (9,29.311) (10,32.121) (11,34.929) (12,37.738) (13,40.544)
        (14,43.351) };
      \addplot coordinates { (1,4.7519) (2,7.7103) (3,10.332)
        (4,12.822) (5,15.254) (6,17.657) (7,20.046) (8,22.429)
        (9,24.808) (10,27.186) (11,29.564) (12,31.940) (13,34.316)
        (14,36.692) };
    \end{axis}
  \end{tikzpicture}
\vspace*{1.3cm}
\end{minipage}
  \caption{Asymptotic of the $F_5$-bound vs linear algebra on the
    Macaulay Matrix}
  \label{fig:BvsGauss}
\end{figure}

\noindent We now draw a few consequences of this theorem.
\paragraph{Numerical estimates}
In view of~\eqref{intro-upper-bound}, $B(\delta)$ can be compared to
$\left({\delta^\delta}/{(\delta-1)^{\delta-1}}\right)^\omega$
where $\omega$ is the exponent of matrix multiplication over $k$, and
therefore our result can be interpreted as a first measure of the
extent to which the $F_5$ algorithm exploits the structure of the
Macaulay matrix for the computation.
In Figure~\ref{fig:BvsGauss}, we display the values of~$B(\delta)$ as
well as the arithmetic complexity of the linear algebra performed on
the Macaulay matrix for different values of the exponent $\omega$.
The first plot gives these values for~$\delta$ from~2 to~10, and the
second one gives the logarithm of these values in terms of
$\log(\delta)$, for $\delta$ from $2^1$ to
$2^{14}=16384$. 
For $2\le\delta<7$, the $F_5$-bound gives a better complexity than
the~\cite{CoWi90} bound~$\omega<2.376$ and the recently improved bounds
down to~$\omega<2.373$ by~\cite{Stothers2010}, \cite{Vassilevska-Williams2011}
and~\cite{Le-Gall2014}. It is better than Strassen's
bound~$\omega=\log_27$ for $2\le \delta<9911$. Thus in practice, for this whole range of degrees, the complexity estimate of $F_5$ behaves asymptotically (wrt to $n$) exponentially better than linear algebra with fast matrix multiplication over the Macaulay matrix.

\paragraph{Nonhomogeneous systems} In the affine case,
Theorem~\ref{thm:nF5} can often be applied with \(\ell=1\):
let $(f_1,\dots,f_n)$ be a system of affine polynomials of identical
degree~$\delta\ge2$ in $k[x_1,\dots,x_n]$; we consider
\((H_1,\ldots,H_n)\) the polynomials obtained by homogenizing the
\(f_{i}\) in \(k[x_1,\dots,x_{n},h]\). Provided $x_1,\dots,x_n,h$ are
in simultaneous Noether position with respect to the system
\((H_1,\ldots,H_n)\) (or equivalently, $x_1,\dots,x_n$ are in
simultaneous Noether position with respect to the system formed by the
homogeneous part of highest degree of the $f_i$), we can then apply the
theorem to \((H_1,\ldots,H_n)\) and derive a bound on the number of operations.
\paragraph{Other term orders} Under the same hypotheses as in Theorem~\ref{thm:nF5}, the computation of a Gröbner basis for the lexicographical order can be achieved by first
computing a Gröbner basis for the grevlex order using Algorithm
matrix-$F_5$ in $O(n B^{n})$ operations and then converting into a
basis for the lexicographical order using the FGLM algorithm
of~\cite{FaGiLaMo93} in~$O(n\delta^{3n})$
operations. Since~$B\ge\delta^3$, the overall complexity is still
bounded by~$O(nB^n)$ arithmetic operations over~$k$.
\paragraph{System solving}
If the field $k$ is infinite and the system is regular, a generic
linear change of variables puts the variables in simultaneous Noether
position. The construction is given for instance
by~Giusti (\cite{Giusti88}, \S5.6), see
also~\cite{GiHe93}. Thus in practice, for zero-dimensional polynomial
system solving, the simultaneous Noether position hypothesis 
can be replaced by the regularity of the system.

\medskip
This article is structured as follows. In Section~\ref{sec1}, we
recall the basic definitions and properties of regular sequences and
Gr\"obner bases, the relation between Gr\"obner bases and linear
algebra and the definition of simultaneous Noether position. In
Section~\ref{sec2}, we give a simple version of the $F_5$
algorithm and we give a structure theorem for Gr\"obner bases computed by this signature-based algorithm. In Section~\ref{sec2bis} we describe more precisely its
behaviour for systems with variables in simultaneous Noether position for the grevlex
ordering and deduce an upper bound for the complexity of $F_5$ in this
case.
Finally, in Section~\ref{sec3} we discuss the practical accuracy of
the bounds we provide in this paper.  In view of our numerical
experiments, the practical behaviour of the algorithm $F_5$ seems to
be asymptotically exponentially better than our bound. A
characterisation of the exact exponent in the complexity is still
open.

Preliminary versions of this work have appeared
 in~Bardet's PhD thesis~(\cite{Bardet04}).
%%%%%%%%%%%%%%%%%%%%%%%%%%%%%%%%%%%%%%%%%%%%%%%%%%%%%%%%%

\section{Gr\"obner Bases and Regularity}\label{sec1}
This section gathers classical definitions and properties, so that this article is self-contained.
We generally follow the terminology and notations
of~\cite{CoLiOS97}.
\subsection{Basic Notation and Definitions}
The polynomial systems we consider are always denoted
$(f_1,\dots,f_m)\in k[x_1,\dots,x_n]$ where $k$ is a field. We denote
by $d_i$ the degree of~$f_i$. 
Throughout this article, the polynomials are homogenous.
The set of homogeneous polynomials of degree~$d$ is denoted~$k[x_1,\dots,x_n]_d$. We use~$\mathcal{T}^i$ to denote the set of nonzero \emph{monomials} in~$x_1,\dots,x_i$ (i.e., products $x_1^{\alpha_1}\dotsm x_i^{\alpha_i}$ with nonnegative integer exponents $\alpha_j$), and $\mathcal{T}^i_d$ the subset of monomials of degree~$d$. When~$i=n$, we use simply~$\mathcal{T}=\mathcal{T}^n$ and $\mathcal{T}_d=\mathcal{T}^n_d.$ The ideal generated by the polynomial system is denoted $\mathcal{I}=\langle f_1,\dots,f_m\rangle$ and the vector space of  homogeneous polynomials of degree~$d$ in~$\mathcal{I}$ is denoted~$\mathcal{I}_d$.

A \emph{monomial ordering} is a total order on monomials that is compatible with the product and such that every nonempty set has a smallest element for the order. Such an ordering is \emph{graded} if monomials of different degrees are ordered according to their degree. The \emph{leading term} $\LT(f)$ of a polynomial $f$ is the \emph{term} (i.e., monomial multiplied by a nonzero constant in $k$) corresponding to its largest monomial for the given monomial ordering.

The \emph{grevlex} ordering is a graded ordering. The order between two monomials of the same degree~$x_\alpha=x_1^{\alpha_1}\dotsm x_n^{\alpha_n}$ and~$x_\beta=x_1^{\beta_1}\dotsm x_n^{\beta_n}$ is given by~$x_\alpha\succ x_\beta$ when the last nonzero element of $(\alpha_1-\beta_1,\dots,\alpha_n-\beta_n)$ is negative. Thus, among the monomials of degree~$d$, the order is
\[x_1^d\succ x_1^{d-1}x_2\succ x_1^{d-2}x_{2}^2\succ\dots\succ x_2^d\succ x_1^{d-1}x_3\succ x_1^{d-2}x_2x_3\succ x_1^{d-2}x_3^2\succ\dots\succ x_n^d.\]

A \emph{Gr\"obner basis} of an ideal~$\mathcal{I}$ for a given monomial ordering is a set~$G$ of generators of~$\mathcal{I}$ such that the leading terms of $G$ generate the monomial ideal $\langle\LT(\mathcal{I})\rangle$, which is the ideal generated by the monomials $\LT(f), f \in \mathcal{I}$. A polynomial is \emph{reduced} with respect to the Gr\"obner basis~$G$ when its leading term is not a multiple of those of~$G$. The basis is \emph{reduced} if each element~$g\in G$ is reduced with respect to~$G\setminus\{g\}$.

\subsection{Macaulay's Matrix}
We recall briefly the construction of this matrix and the explicit
relation with Gr\"obner bases, which can be found in the works
by~\cite{Lazard83} and~\cite{Giusti84,Giusti85}. There are several
advantages to this point of view: one is that, for a graded order, it
gives an easy access to the Hilbert function of the ideal, another one
is that upper bounds on the complexity are easily recovered from
classical linear algebra.

For a given degree~$d$ and polynomial system~$(f_1,\dots,f_m)$, Macaulay's matrix~$\Mac{d}{m}$ has its columns indexed by the monomials of~$\mathcal{T}_d$. For each polynomial~$f_i$ of the system and each monomial~$t\in\mathcal{T}_{d-d_i}$, it contains one row whose entry in the column indexed by a monomial $t'$ is the coefficient of~$t'$ in $tf_i$. Thus, the rows of this matrix generate the vector space~$\mathcal{I}_d$.
The \emph{Hilbert function} is defined by 
\[\HF_\mathcal{I}(d)=\dim k[x_1,\dots,x_n]_d/\mathcal{I}_d,\]
it is therefore equal to the dimension $\binom{n+d-1}{d}$ of $k[x_1,\dots,x_n]_d$ minus the rank of $\Mac{d}{m}$. For~$d$ large enough, this function is a polynomial (the \emph{Hilbert polynomial}~$\operatorname{HP}_\mathcal{I}$). 
The generating function~$\operatorname{H}_\mathcal{I}=\sum_{d\ge0}{\HF_\mathcal{I}(d)z^d}$ is called the \emph{Hilbert series} of the ideal.

\subsection{Gr\"obner Bases}

Performing a Gaussian elimination on the matrix~$\Mac{d}{m}$ computes a basis of~$\mathcal{I}_d$. If moreover, the columns are ordered by decreasing order with respect to the chosen graded monomial ordering, and column pivoting is not allowed, then the leading terms of this basis give~$\LT(\mathcal{I}_d)$. From these bases for $\min(d_1,\dots,d_m)\le d\le D$, where $D$ is the maximal degree of the elements in the reduced Gr\"obner basis of~$\mathcal{I}$, this Gr\"obner basis can be reconstructed. This way, the computation is reduced to linear algebra operations. 

From there we  now prove the general upper bound on the complexity that has been given in Proposition~\ref{propupper}.

Macaulay's matrix $\Mac{d}{m}$ has $C_d=|\mathcal{T}_d|=\binom{n+d-1}{d}$ columns and $R_d=|\mathcal{T}_{d-d_1}|+\dots+|\mathcal{T}_{d-d_m}|$ rows. A basis of its rows is obtained by the computation of a reduced row echelon form. \cite{Storjohann00} has shown that fast matrix multiplication can be used to compute this form with a complexity of~$O(R_dC_dr^{\omega-2})$, where $r$ is the rank of the matrix. Thus, this is bounded by~$O(R_dC_d^{\omega-1})$.

The computation is performed for~$d=\min(d_1,\dots,d_m),\dots,D$. The number of rows is bounded by $mC_d$ and $C_d$ is an increasing sequence, so that the conclusion of Proposition~\ref{propupper} follows.

The remaining problems are to bound~$D$ and to perform this Gaussian elimination efficiently by taking the structure of the matrix into account. An important class where this can be done is the class of regular systems.
%%%%%%%%%%%%%%%%%%%%%%%%%%%%%%%%%%%%%%%%%%%%%%%%
\subsection{Regular Systems}
\begin{definition}\label{defregsys}
 $(f_1,\ldots,f_m)$ is regular if for all $i=1,\ldots,m$, $f_i$ is
 not a zero-divisor in the quotient ring
 $k[x_1,\ldots,x_n]/\langle f_1,\ldots,f_{i-1}\rangle$. In other words if there
 exists $g$ such that
$
g f_i \in \langle f_1,\ldots,f_{i-1}\rangle
$
then $g$ belongs to $\langle f_1,\ldots,f_{i-1}\rangle$.
\end{definition}
The following Lemma gives a characterisation of zero-divisors for
homogeneous polynomials:
\begin{lemma}\label{lemma:zero-div}
  Let $\mathcal I\subset k[x_1,\dots,x_n]$ be an homogeneous ideal,
  $\mathcal I \neq \langle 1 \rangle$, and $f \in k[x_1,\dots,x_n]$
  homogeneous of degree $\delta\ge 1$. Then $f$ is not a zero-divisor
  in $k[x_1,\dots,x_n]/\mathcal I$ if and only if $
  \operatorname{H}_{\mathcal I+ \langle f \rangle}(z) =
  (1-z^\delta)\operatorname{H}_{\mathcal{I}}(z).$
\end{lemma}
\begin{proof}
The proof boils down to using the relation between the dimensions of the kernel $K_d$ and image of the application of multiplication by~$f$ from $k[x_1,\dots,x_n]_{d-\delta}/\mathcal{I}_{d-\delta}$ to $k[x_1,\ldots,x_n]_d/\mathcal{I}_d$. 
This gives
\begin{equation}\label{HF-identity}
\HF_{\mathcal I + \langle f\rangle}(d)=
\HF_{\mathcal I}(d)
-\HF_{\mathcal I}(d-\delta) + \dim(K_d),\qquad d\in \mathbb{N}.
\end{equation}
(with the convention $\HF(-d)=0$ for $d\ge 1$). 
Multiplying~\eqref{HF-identity} by~$z^d$ and summing over~$d$ leads to 
\[H_{\mathcal I + \langle f\rangle}(z)=(1-z^{\delta})H_{\mathcal I}(z)
+ \sum_{d\ge 0} \dim(K_d)z^d.\]
The equivalence results from the
definition: $f$ is not a zero-divisor in
$k[x_1,\dots,x_n]/\mathcal{I}$ if and only if $\dim(K_d)=0$ for all
$d\ge 0$.
\end{proof}
This leads to a classical property of regular systems, essentially due
to~Macaulay (\cite[\S58]{Macaulay16}):
\begin{proposition}\label{prop:hilbertseries} The system of
  homogeneous polynomials $(f_1,\ldots,f_m)\subset k[x_1,\dots,x_n]$
  is regular if and only if its Hilbert series is
\begin{equation}\label{serie}
\operatorname H_{\mathcal{I}}(z) = \frac{\Pi_{j=1}^m(1-z^{d_j})}{(1-z)^n}.
\end{equation}
If $m=n$, then the sequence $(f_1,\dots, f_n)$ is regular if and only
if its Hilbert series is a polynomial.
\end{proposition}
\begin{proof}
  The first part follows from the previous lemma
  and~$H_{\langle\rangle}(z)=(1-z)^{-n}$ that counts the number of
  homogeneous monomials. 

  For $m=n$, if $\operatorname H_{\mathcal I}(z)$ is a polynomial,
  then Bézout's bound~\cite[ch.~3\S5]{CoxLittleOShea2005} states that $\operatorname H_{\mathcal
    I}(1)$, which is the number of solutions of $\mathcal I$ in the
  algebraic closure of $k$, is bounded by $\Pi_{j=1}^n d_j$.  But
  Equation~\eqref{HF-identity} leads to $\operatorname
  H_{\mathcal{I}}(z) \ge \prod_{j=1}^n{(1-z^{d_j})}/{(1-z)^n}$ with
  inequality coefficient by coefficient. By taking the value at $z=1$,
  we get $\prod_{j=1}^n d_j \ge \operatorname H_{\mathcal I}(1) \ge
  \prod_{j=1}^n d_j$, which proves the equality. As each coefficient of
  the Hilbert series is nonnegative, we deduce that $\operatorname
  H_{\mathcal{I}}(z) = \prod_{j=1}^n{(1-z^{d_j})}/{(1-z)^n}$ and the
  sequence is regular by the first part of the lemma.
\end{proof}
\begin{corollary}\label{cor:m=n}Let $(f_1,\dots,f_n)$ be a regular
  system of homogeneous polynomials in $k[x_1,\dots,x_n]$, then the highest degree in the
  elements of a Gr\"obner basis for a graded ordering is bounded by
  Macaulay's bound~\eqref{mac_bound}.
\end{corollary}
\begin{proof} When~$m=n$, the Hilbert series~\eqref{serie} is a polynomial, whose degree $D$ is one less than the bound~\eqref{mac_bound}. This implies that the Hilbert function is~0 for degree~$D$. In other words all the monomials of~$\mathcal{T}_D$ belong to~$\mathcal{I}_D$, whence the result.
\end{proof}
\begin{example}\label{ex:SNP}
The system 
  \begin{equation}
 \left\{ \begin{array}{rcl}
    f_1 &=& x^2 + y^2 -2xz -2yz + z^2 + h^2\\
    f_2 &=& x^2+x y+ y z - z^2 - 2h^2\\
    f_3 &=&  x^2 - y^2 + 2yz - 2 z^2 \\
  \end{array}
  \right.
\label{eq:circles}
  \end{equation} 
  in $k[x,y,z,h]$, represents, for $h=0$, the intersection of a
  projective circle and two hyperbolas over $k=\mathbb R$. The
  coefficient of $h^2$ is chosen so that the point $(1,1,1,1)$ is
  a solution of the system. The Hilbert series of the systems
  $(f_1,f_2)$ and $(f_1,f_2,f_3)$ (computed from a Gröbner Basis for a
  grevlex ordering) are respectively $H_{(f_1,f_2)}(t) =
  {(1+t)^2}/({1-t})^2=(1-t^2)^2/(1-t)^4$ and $H_{(f_1,f_2,f_3)}(t) =
  {(1+t)^3}/{(1-t)}=(1-t^2)^3/(1-t)^4$, showing that these systems are
  regular.
\end{example}

%%%%%%%%%%%%%%%%%%%%%%%%%%%%%%%%%%%%%%%%%%%%%%%%
\subsection{Noether position}\label{noether}
\begin{definition} The variables $(x_1,\dots,x_m)$ are in \emph{Noether position} with respect to the system $(f_1,\dots,f_m)$ if their canonical images in $k[x_1,\dots,x_n]/\langle f_1,\dots,f_m\rangle$ are algebraic integers over $k[x_{m+1},\dots,x_n]$ and moreover $k[x_{m+1},\dots,x_n]\cap\langle f_1,\dots,f_m\rangle=\langle 0\rangle$.
\end{definition}
The variable $x_i\in k[x_1,\dots,x_n]/\langle f_1,\dots,f_m\rangle$ is
an algebraic integer over $k[x_{m+1},\dots,x_n]$ when there exists a
polynomial $g\in k[x_i,x_{m+1},\dots,x_n]\cap \langle
f_1,\dots,f_m\rangle$ that is monic with respect to~$x_i$. A Gröbner
basis of $\langle f_1,\dots,f_m\rangle$ for an elimination monomial
ordering such that $\{x_j, 1\le j\neq i\le m\}>\{x_i,
x_{m+1},\dots,x_n\}$ contains such a $g$ (up to a constant) if and
only if $x_i$ has the desired property.
\begin{example}
  The variables $(x,y)$ are in Noether position with respect to the
  system $(f_1,f_2)$ from Example~\ref{ex:SNP}.  Indeed, a Gröbner
  basis of the system for the lexicographical ordering $x>y>z>h$
  (resp. $y>x>z>h$) contains the polynomial
  $2\,{y}^{4}-6\,{y}^{3}z+12\,{y}^{2}{z}^{2}+7\,{y}^{2}{h}^{2}-8\,y{z}^{3}-20\,yz{h}^{2}+4\,{z}^{2}{h}^{2}+9\,{h}^{4}$
  (resp. $2\,{x}^{4}+2\,{x}^{3}z-2\,{x}^{2}{z}^{2}-3\,{x}^{2}{h}^{2}-2\,x{z}^{3}-2\,xz{h}^{2}+{z}^{2}{h}^{2}+4\,{h}^{4}
  $).

  On the other hand, the variables $(y,z)$ are not in Noether position
  with respect to the system $(f_1,f_2)$. Again, a Gröbner basis
  computation for the lexicographical ordering shows that $\mathcal I
  \cap k[y,x,h] = \langle
  2\,{y}^{3}x+4\,{y}^{2}{x}^{2}-{y}^{2}{h}^{2}+8\,y{x}^{3}-8\,yx{h}^{2}-4\,{x}^{2}{h}^{2}-{h}^{4}
  \rangle$.
\end{example}
Geometrically, the  Noether position implies that the algebraic set defined
by the system has dimension~$n-m$ and, in an algebraic closure of~$k$,
for any value of $(x_{m+1},\dots,x_n)$, the system has exactly the
same number of solutions (counting multiplicity). In a sufficiently large field, 
for regular systems, the variables can be put in Noether position by a generic linear change of variables, as explained by \cite{Giusti88}.
 
The following proposition characterises algebraically the Noether
position property for homogeneous ideals. It shows that in this case, this position is also
equivalent to the condition that the system
$(f_1,\dots,f_m,x_{m+1},\dots,x_n)$ has only the solution $\{0\}$.
\begin{proposition}[\cite{LJ84}]\label{prop6}%, Ch.~2, Prop.~5.4]
  Let $(f_1,\dots,f_m)$ be a system of homogeneous polynomials of
  $k[x_1,\dots,x_n]$, such that $ \langle f_1,\dots,f_m\rangle \neq
  \langle 1 \rangle$. If the variables $(x_1,\dots,x_m)$ are in
  {Noether position} with respect to the system~$(f_1,\dots,f_m)$, then
  the sequence $(f_1,\ldots,f_m,x_{m+1},\ldots,x_n)$ is
  regular.
\end{proposition}
\begin{proof}
  The variables $(x_1,\dots,x_m)$ being in {Noether position} with
  respect to the system~$(f_1,\dots,f_m)$, for each $1\le
  i \le m$, there exists a polynomial $g\in
  k[x_i,x_{m+1},\dots,x_n]\cap \langle f_1,\dots,f_m\rangle$ of degree
  $n_i\ge 1$ in $x_i$ such that the coefficient of $x_i^{n_i}$ in $g$ is
  1. 
This implies
  that the Hilbert series
  $H_{(f_1,\dots,f_m,x_{m+1},\dots,x_n)}(z)$ is a polynomial.  Then by
  Proposition~\ref{prop:hilbertseries}, 
  $(f_1,\dots,f_m,x_{m+1},\dots,x_n)$ is a regular sequence.
\end{proof}

From the computational point of view, the following proposition gives
very precise information on the structure of a grevlex Gr\"obner
basis, it plays an important role in obtaining good complexity estimates
in the next section.
\begin{proposition}[\cite{LJ84}, Ch.~3, Prop.~3.4]\label{prop:LJ}
Let $(x_1,\dots,x_m)$ be in 
Noether position with respect to the homogeneous system~$(f_1,\dots,f_m)$. Let $\theta_m$ be a ring endomorphism of $k[x_1,\dots,x_n]$ such that $\theta_m(x_i)=x_i$ for $i\in\{1,\dots,m\}$, while $\theta_m(x_i)=0$ for $i>m$. Then, for the grevlex monomial ordering
\[\LT(\langle f_1,\dots,f_m\rangle)=\LT(\theta_m(\langle f_1,\dots,f_m\rangle))\cdot\langle x_{m+1},\dots,x_n\rangle.\]
\end{proposition}
In other words, the leading terms of the elements of the reduced Gr\"obner basis do not depend on the variables~$(x_{m+1},\dots,x_n)$.

\begin{proof}
  Let $\mathcal{I}=\langle f_1,\dots,f_m\rangle$. The inclusion
  $\LT(\mathcal{I}) \supset \LT(\theta_m(\mathcal{I}))\cdot\langle
  x_{m+1},\dots,x_n\rangle$ follows from the fact that for the grevlex
  monomial ordering, when~$\theta_m(f)\neq0$,
  $\LT(f)=\LT(\theta_m(f))$.

  Conversely let~$f\in\mathcal{I}$ and let~$M=x_1^{\alpha_1}\dotsm
  x_n^{\alpha_n}$ be its leading monomial for the grevlex ordering. We
  have to prove that there exists~$g\in\mathcal{I}$ with leading
  monomial~$x_1^{\alpha_1}\dotsm x_m^{\alpha_m}$.  Since the ideal is
  homogeneous, we can assume $f$ to be homogeneous as well.  Let~$l$
  be the largest index such that~$x_l|M$. By definition of the grevlex
  ordering and the fact that $M$ is the leading monomial of $f$, there
  exist homogeneous polynomials $g_l, \dots, g_n\in k[x_1,\dots,x_n]$ 
  such that
  \begin{equation}
  f = x_l^{\alpha_l} g_l+ x_{l+1}g_{l+1}+\dots+x_ng_n,\qquad \text{$g_l\in
  k[x_1,\dots,x_{l}]\backslash \{0\}$ 
 and  $\LT(g_l) = x_1^{\alpha_1}\dotsm x_{l-1}^{\alpha_{l-1}}$.}
\end{equation}
By Proposition~\ref{prop6}, the sequence $(f_1,\dots,
f_m,x_{m+1},\dots,x_n)$ is regular.  If~$l>m$, then \(f\equiv x_l^{\alpha_l} g_l\equiv0\mod \mathcal{I}+\langle
x_{l+1},\dots,x_{n}\rangle\) and since from
Proposition~\ref{prop:hilbertseries}  $x_l$ is not a zero-divisor
in~$k[x_1,\dots,x_n]/(\mathcal{I}+\langle
x_{l+1},\dots,x_{n}\rangle)$ we deduce successively that  \(g_l\equiv0\mod \mathcal{I}+\langle
x_{l+1},\dots,x_{n}\rangle\) and  \(g_l\equiv0\mod \mathcal{I}\).
  Hence, starting from \(f\in\mathcal{I}\) such that
   \(\LT(f)\in k[x_1,\ldots,x_l]\) with \(l>m\), we obtain \(g_l \in\mathcal{I}\)  such that \(\LT(f)=x_l^{\alpha_l} \LT(g_l)\) and \(\LT(g_l)\in k[x_1,\ldots,x_{l-1}]\). By induction on \(l\) we can find a polynomial \(g\in\mathcal{I}\) such that \(\LT(f)=x_{m+1}^{\alpha_{m+1}}\dotsm x_l^{\alpha_l}\LT(g)\) and \(\LT(g)\in k[x_1,\ldots,x_{m}]\). This proves the converse inclusion.
\end{proof}

In view of the incremental nature of the $F_5$ algorithm, an even stronger property will be useful in our considerations.
\begin{definition}
The variables $(x_1,\dots,x_n)$ are in \emph{simultaneous Noether position} with respect to the system $(f_1,\dots,f_m)$ when the variables $(x_1,\dots,x_i)$ are in Noether position with respect to $(f_1,\dots,f_i)$ for all $i\in\{1,\dots,m\}$.
\end{definition}

\begin{example}
  The variables $(x,y,z,h)$ in Example~\ref{ex:SNP} are in
  simultaneous Noether position with respect to the system
  $(f_1,f_2,f_3)$ from Equation~\eqref{eq:circles}.
\end{example}
Again, this situation is generic for regular systems and can be reached by a linear change of coordinates if the field is sufficiently large.

%%%%%%%%%%%%%%%%%%%%%%%%%%%%%%%%%%%%%%%%%%%%%%%%
\section{Signature-based Gröbner basis computations: the $F_5$ Algorithm}\label{sec2}
%%%%%%%%%%%%%%%%%%%%%%%%%%%%%%%%%%%%%%%%%%%%%%%%
\subsection{Description of the Matrix-$F_5$ algorithm}
\label{sec:desc_F5}

Faugere's  $F_5$ algorithm (\cite{Faugere02}) is designed so that it ensures that no ``useless'' reduction to~$0$ is performed when the input system is regular.
We now describe a matrix version of $F_5$ that is
well-suited to a complexity analysis. The main difference with the original algorithm is that the maximal degree occurring in the computation is given as an input of the algorithm.
As in Faugere's $F_4$ algorithm  (\cite{Faugere99}), linear algebra is used to reduce the polynomials. The resulting algorithm is very easy to implement. It is probably somewhat less efficient than the original~$F_5$ on most practical examples, but it lets us compute an upper bound on the complexity of $F_5$. It is this matrix variant that was used with success by~\cite{FaJo03}.

In order to keep track of the polynomials that lead to the different rows of the matrices encountered during the algorithm, it is convenient to view a matrix $(M)$ as a map $(s,t)\in S\times T\mapsto M_{s,t}\in k$
where $S$ is a finite subset of $\mathbb{N}\times {\mathcal T}$ and $T$ a finite subset of ${\mathcal T}$
ordered  using a graded ordering. A row indexed by~$s=(i,\tau)$ will be used to represent a polynomial obtained as the sum of $\tau f_i$ and some other ``smaller'' polynomials in~$\mathcal{I}$; this index~$s$ is the \emph{signature} of the corresponding polynomial. A row in the matrix $M$ is specified by its signature $s$, and we identify the vector $\operatorname{Row}(M,s) = [M_{s,t}\,|\,t\in T]$ and the polynomial $\sum_{t\in T}
M_{s,t} t$; the leading term of a row is the leading term of the corresponding
polynomial. We fix the following notation: $\operatorname{Rows}(M)=S$ and $\LT(M)$ is the set of leading terms of
all the rows of $M$. A \emph{valid} elementary row operation on $M$ consists in replacing the row
$s\in S$ by the linear combination $\operatorname{Row}(M,s) \leftarrow \operatorname{Row}(M,s) +
\lambda \operatorname{Row}(M,s')$ where $\lambda \in k$, $s'\in S$ and the {\em additional condition} that $s' =
(j',u') < s=(j,u)$ (i.e., $j'<j$ or ($j=j'$ and $u'\prec u$)). The index of the line is unchanged. 
We denote by $\tilde{\mathcal{M}}_{d,i}$ the result of Gaussian elimination applied to the matrix
$\mathcal{M}_{d,i}$ using a sequence of \emph{valid} elementary row operations. 

There are two distinct ways of
performing a valid Gaussian elimination: either we perform reductions
only for the leading term of each row, in which case we call the
reduction a \emph{top-reduction}, or we perform more valid reductions
so that each column containing a leading coefficient has zeros
elsewhere below, in which case we call the reduction a
\emph{full-reduction}. The complexity analysis of this paper is done
in the top-reduction case, and in Section~\ref{sec3} an experimental
comparison with the full-reduction case is given.

The algorithm matrix-$F_5$ constructs matrices incrementally in the degree and the number of
polynomials. Let $d$ be the current degree and $i$ the current number of polynomials (in other words
we are computing a Gr\"obner basis of $\langle f_1,\ldots,f_i\rangle$ truncated in degree $d$).  The algorithm constructs a
matrix $\mathcal{M}_{d,i}$ obtained from the  Macaulay matrix $\Mac d i$ by removing selected rows.
With the previous notation, $\mathcal{M}_{d,i}$ is a map $S\times
\mathcal{T}_d\mapsto k$ such that $S$ is a subset of $\{1,\ldots,i\}\times{\mathcal T}$.

\cite{Faugere02} defines the signature of a polynomial and uses it to give a new criterion to
remove useless computations. In the matrix-$F_5$ algorithm the signatures become the indices 
of the rows and the original criterion translates as:
\begin{proposition}[$F_5$ criterion]\label{F5-crit}
If $t$ is the leading term of $\operatorname{Row}(\tilde{\mathcal{M}}_{d-d_i,i-1},s)$ where
$s<(i,1)$ then the row indexed by $(i,t)$ 
belongs to the vector space generated by the rows of $\Mac{d}{i}$ having smaller index.
\end{proposition}
\begin{proof}The hypothesis is that $t\in\LT(\langle f_1,\dots,f_{i-1}\rangle_{d-d_i})$, say $t=\LT(h)$ with $h=\sum_{k=1}^{i-1}h_k f_k$. This implies that $t f_i = \sum_{k=1}^{i-1}  f_i h_kf_k + (t-h) f_i$, where the first term belongs to $\langle \operatorname{Row}\Mac{d}{i-1} \rangle$ and the last one is a linear combination of rows of $\Mac{d}{i}$ having smaller index, as $\LT(t-h) \prec \LT(h)$.
\end{proof}

We now describe the matrix-$F_5$ algorithm. Here the order is any monomial ordering. It enters the algorithm through the function~$\LT$.
\begin{description}
\item[Algorithm matrix-$F_5$]\ \\
Input: homogeneous polynomials $(f_1,\dots,f_m)$ with degrees $d_1\leq\cdots\leq d_m$;\\
\phantom{Input:} a maximal degree $D$. \\
Output: The elements of degree at most~$D$ of the reduced\\
\phantom{Input:} Gr\"obner bases of $(f_1,\dots,f_i)$, for $i=1,\dots,m.$ \\
\ligne0 {\bf for} $i$ {\bf from} 1 {\bf to} $n$ {\bf do}
$G_i:=\emptyset$; \textbf{end for} // \emph{initialise the Gröbner Bases}
$G_i$ of $(f_1,\dots,f_i)$.
\\
\ligne0 {\bf for} $d$ {\bf from} $d_1$ {\bf to} $D$ {\bf do}\\
\ligne5   $\mathcal{M}_{d,0} := \emptyset$, $\tilde{\mathcal{M}}_{d,0} := \emptyset$\\
\ligne5   {\bf for} $i$ {\bf from} $1$ {\bf to} $m$ {\bf do}\\
\ligne{10}    {\bf if} $d<d_i$ {\bf then} ${\mathcal{M}}_{d,i} := {\mathcal{M}_{d,i-1}}$\\
\ligne{10}    {\bf else if} $d=d_i$ {\bf then}\\
\ligne{15}      $\mathcal{M}_{d_i,i}$ := add the new row $f_i$ to $\tilde{\mathcal{M}}_{d_i,i-1}$  with index $(i,1)$\\
\ligne{10}    {\bf else}\\
\ligne{15}      $\mathcal{M}_{d,i} := \tilde{\mathcal{M}}_{d,i-1}$ \\
\ligne[\label{ligne:10}]{15}       Crit $ := \LT(\tilde{\mathcal{M}}_{d-d_i,i-1})$ \\
\ligne{15}      {\bf for} $f$ {\bf in} $\operatorname{Rows}(\mathcal{M}_{d-1,i})\backslash \operatorname{Rows}(\mathcal{M}_{d-1,i-1})$ {\bf do}\\
\ligne{20}         $(i,u):=\operatorname{index}(f)$, with $u=x_{j_1} \cdots x_{j_{d-d_i-1}},$\\
\ligne{50}             and  $1\leq j_1 \leq \cdots \leq j_{d-d_i-1}\leq n$\\
\ligne[\label{ligne:1}]{20}         {\bf for} $j$ {\bf from} $j_{d-d_i-1}$ {\bf to} $n$ {\bf do}\\
\ligne{25}            {\bf if} $u x_j\not\in\operatorname{Crit}$ {\bf then}\\
\ligne[\label{ligne:16}]{30}               add the new row $x_j f$ with index $(i, u x_j)$ in $\mathcal{M}_{d,i}$\\
\ligne{25}            {\bf end if}\\
\ligne{20}         {\bf end for}\\
\ligne{15}      {\bf end for}\\
\ligne{10}      {\bf end if}\\
\ligne{10}    Compute $\tilde{\mathcal{M}}_{d,i}$ by Gaussian elimination from ${\mathcal{M}}_{d,i}$\\
\ligne[\label{ligne:22}]{10}    Add to $G_i$ all rows of $\tilde{\mathcal{M}}_{d,i}$ not reducible by $\LT(G_i)$\\
\ligne{5}      {\bf end for}\\
\ligne0       {\bf end for}\\
\ligne0 {\bf return} $[G_i\,|\, i=1,\ldots,m]$
\end{description}
The \textbf{for} loop of
line~\ref{ligne:1} constructs the matrix $\mathcal M_{d,i}$
containing all the polynomials $x_1^{\alpha_1}\dotsm
x_n^{\alpha_n}f_i$ with $\alpha_1+\dotsm +\alpha_n = d-d_i$ (except
some that reduce trivially to zero). In order to avoid redundant computations,
these are constructed from the rows of the previous matrix $\mathcal
M_{d-1,i}$ by multiplying all rows by all variables. A row indexed by
$(i,x_1^{\alpha_1}\dotsm x_j^{\alpha_j})$ with $\alpha_j\neq 0$ can
arise from several rows in $\mathcal M_{d-1,i}$, we choose to
construct it from the row indexed by $(i,u)$ in $\mathcal M_{d-1,i}$
with $u=x_1^{\alpha_1}\dotsm x_j^{\alpha_j-1}$ and multiply it by
$x_j$, the largest variable occurring in $u$. This insures
that every row comes from exactly one row in the previous matrix.

\begin{example}\label{ex:matrices}
  Algorithm matrix-$F_5$ over Example~\ref{ex:SNP} constructs the
  following matrices. In degree~2,
  \[\mathcal M_{2,3} = \bordermatrix{ & {x}^{2} & xy & {y}^{2} & xz &
    yz & {z}^{2} & hx & yh & zh & {h}^{2}\cr f_1 & 1 & & 1 & -2 & -2 &
    1 & & & & 1 \cr f_2 & 1 & 1 & & & 1 & -1 & & & & -2 \cr f_3 & 1 &
    & -1 & & 2 & -2 \cr },\]
where only nonzero entries are displayed. Gaussian reduction yields
  \[\tilde{\mathcal M_{2,3}} = \bordermatrix{ & {x}^{2} & xy & {y}^{2}
    & xz & yz & {z}^{2} & hx & yh & zh & {h}^{2}\cr f_1 & 1 & & 1 & -2
    & -2 & 1 & & & & 1 \cr f_2 & & 1 & -1 & 2 & 3 & -2 & & & & -3 \cr
    f_3 & & & 2 & -2 & -4 & 3 & & & & 1 \cr }.\]
From $\tilde{\mathcal M_{2,3}}$, it follows that the indices and leading terms of the elements in $G_3$ are
  \[\{((1,1), x^2), ((2,1), xy), ((3,1), y^2)\}.\]
  Next, in degree 3, $\tilde{\mathcal M_{3,3}}$ contains (in that order) the columns 
\[{x}^{3},{x}^{2}y,{y}^{2}x,{y}^{3},{x}^{2}z,zxy,z{y}^{2},x{z}^{2},y{z}^{2},{z}^{3},{x}^{2}h,xyh,{y}^{2}h,xzh,yzh,{z}^{2}h,{h}^{2}x,{h}^{2}y,{h}^{2}z,{h}^{3}\]
  and the rows with indices and leading terms
  \[ \begin{array}[h]{@{}l@{~}llllllllllll} (\text{ind.})& (1,h)&
    (1,z)& (1,y)&
    (1,x)& (2,h)& (2,z)& (2,y)& (2,x)& (3,h)& (3,z)& (3,y)& (3,x)\\
    (\LT)& x^2h & x^2z & x^2y & x^3& xyh& xyz& xy^2& \underline{y^3}&
    y^2h& y^2z& \underline{xz^2}& \underline{yz^2}
\end{array}
\]
The underlined leading terms are those inserted into $G_3$ by Algorithm
matrix-$F_5$ in line~\eqref{ligne:22}. Degree~4 is the first time the $F_5$ criterion is used. The set Crit of
line~\eqref{ligne:10} is empty for $i=1$ by convention, but it contains $x^2$ for $i=2$ and
$x^2,xy$ for $i=3$.  Thus in line~\eqref{ligne:16},
all rows $(i,m)$ with $i = 1, 2, 3$ and $m$ a monomial of degree
2 are added to $\mathcal M_{4,2}$ and $\mathcal M_{4,3}$, except the rows $(2,x^2)$, $(3, xy)$ and $(3,x^2)$. The matrix
$\tilde{\mathcal M_{4,3}}$ contains $\binom{7}{4} = 35$ columns and
$3\binom52-3 = 27$ rows. The rows that are reduced during the Gaussian
elimination and that are added to the Gröbner basis are $((3,y^2),
z^4)$. No reduction to~0 has occurred and the Gröbner bases are
\[G_1 = \{((1,1), x^2 + y^2 -2xz -2yz + z^2 + h^2)\},\]
\[G_2 = G_1 \cup
\left\{
    \begin{array}{ll}
((2,1), &xy-{y}^{2}+2\,xz+3\,yz-2\,{z}^{2}-3\,{h}^{2})\\
((2,x), & 2\,{y}^{3}-7\,xyz-3\,{y}^{2}z-2\,x{z}^{2}-y{z}^{2}+2\,{z}^{3}+3\,x{h}^{2}+4\,y{h}^{2}+2\,z{h}^{2})
\end{array}
\right\}, \]
\[G_3 = G_2 \cup
\left\{
    \begin{array}{ll}
((3,1), & 2\,{y}^{2}-2\,xz-4\,yz+3\,{z}^{2}+{h}^{2})\\
((3,y), & 4\,x{z}^{2}+3\,y{z}^{2}-2\,{z}^{3}+3\,x{h}^{2}+3\,y{h}^{2}-11\,z{h}^{2}) \\
((3,x), & 3\,y{z}^{2}-6\,{z}^{3}+11\,x{h}^{2}-5\,y{h}^{2}-3\,z{h}^{2}) \\
((3,y^2), & 3\,{z}^{4}+4\,xz{h}^{2}+12\,yz{h}^{2}-7\,{z}^{2}{h}^{2}-12\,{h}^{4}) \\
\end{array}
\right\}. \]

\end{example}

The main property of this algorithm is given in the following theorem.
\begin{theorem}\label{thmf5}
The algorithm matrix-$F_5$ computes the elements of degree at most~$D$ of the reduced Gr\"obner bases of $\langle f_1,\ldots,f_i\rangle$, $i=1,\dots,m$.
Moreover if $(f_1,\ldots,f_m)$ is a regular sequence then all the matrices $\mathcal{M}_{d,i}$ have full rank.
\end{theorem}
\begin{proof}
The proof of the first statement follows the algorithm: it is an induction on~$d$ and~$i$. For $d=d_1$ and $i=1$, the result is clear. Assuming the induction hypothesis, we now have to prove that the rows of $\mathcal{M}_{d,i}$ generate~$\langle f_1,\dots,f_i\rangle_d$. Then we can deduce that $\LT(\tilde{\mathcal{M}}_{d,i})$ generates~$\LT(\langle f_1,\dots,f_i\rangle_d)$ and the conclusion on~$G_i$ follows.

It is thus sufficient to show that for any $\tau\in\mathcal{T}_{d-d_i}$, the polynomial~$\tau f_i$ is generated by the rows of~$\mathcal{M}_{d,i}$. If $d\le d_i$ the result is clear. Otherwise, let $j$ be the highest index such that $x_j\mid\tau$ and let $u=\tau/x_j$. If $u$ is not the index of a row of $\mathcal{M}_{d-1,i}\setminus\mathcal{M}_{d-1,i-1}$ then by  the induction hypothesis $uf_i$ is generated by the rows of~$\mathcal{M}_{d-1,i-1}$ and therefore $\tau f_i$ is generated by the rows of $\mathcal{M}_{d,i-1}$. This justifies the selection of rows in the loop over $s$. Otherwise, $\tau f_i$ is entered by the algorithm in~$\mathcal{M}_{d,i}$, unless $\tau\in\LT(\tilde{\mathcal{M}}_{d-d_i,i-1})$ since then~$\tau$ is eliminated by the criterion, thanks to  Proposition~\ref{F5-crit}. 

The second part of the theorem is proved by contradiction.
If a row of~$\mathcal{M}_{d,i}$ indexed by~$(i,u)$ reduces to~0, this means that the algorithm has constructed an identity
\[g_if_i+\dots+g_1f_1=0.\]
Moreover, the criterion ensures that $g_i\neq0$ is reduced with respect to $\langle f_1,\dots,f_{i-1}\rangle$. This contradicts the regularity of $(f_1,\dots,f_i)$. 
\end{proof}

Other useful properties of the algorithm matrix-$F_5$ that are needed later are gathered in the following Lemma.
\begin{lemma}\label{lemmaf5}
        \begin{enumerate}
                \item Any row entered by the algorithm matrix-$F_5$ into the matrix~$\mathcal{M}_{d,i}$ represents a polynomial
\[g_if_i+\dots+g_1f_1,\]
where~$g_i$ is reduced with respect to~$\langle f_1,\dots,f_{i-1}\rangle$;
\item if $g\in G_i\setminus G_{i-1}$, then its index $s_g$ has the form~$(i,t)$.
\end{enumerate}
\end{lemma}
\begin{proof} The first property comes from the proof of the previous theorem.
The second one comes from the fact that the algorithm works incrementally with respect to~$i$.
\end{proof}

%%%%%%%%%%%%%%%%%%%%%%%%%%%%%%%%%%%%%%%%%%%%%%%%

\subsection{Structure theorem for Grevlex Bases with variables in Simultaneous Noether Position}
The estimate in Proposition~\ref{propupper} describes precisely the shape of the final Gr\"obner basis, but it does not take into account
any specificity of the $F_5$ algorithm. Hence we first study in more detail the structure of
grevlex bases computed by \(F_{5}\), giving the shape of the signatures associated to those polynomials. We further restrict to a special
situation, namely when the variables are in \emph{simultaneous Noether
  position}.
The following proposition will play a crucial role when estimating precisely the number of operations of the matrix-$F_5$ algorithm.
\begin{proposition}\label{cor:structure} [Structure of the $F_5$-bases] Let $(f_1,\dots,f_m)$ be a
  homogeneous system for which the variables~$(x_1,\dots,x_n)$ are in
  simultaneous Noether position. Let $G_1,\dots,G_m$ be the result of
  the matrix-$F_5$ algorithm applied to this system for the grevlex
  ordering. Then, for all $((j,t),g)\in G_i$, one has $j\le i$,
  $\LT(g)\in\mathcal{T}^j$ and $t\in\mathcal{T}^{j-1}$.
\end{proposition}
\begin{proof}
That~$j\le i$ is a consequence of the incremental nature of the algorithm. Also, since reductions do not change the index and only involve rows with smaller indices, one has that if~$((j,t),g)\in G_i$, then~$((j,t),g)\in G_j$. Thus by induction it is sufficient to consider~$((i,t),g)$. The result on~$\LT(g)$ is given by Proposition~\ref{prop:LJ} using the simultaneous Noether position hypothesis.
        
The property that~$t\in\mathcal{T}^{i-1}$ is more deeply related to
the way the algorithm~$F_5$ works. We prove it by induction on~$i$ and
on the degree of~$t$. For~$i=1$, the only element of the basis
is~$((1,1),f_1)$; it satisfies the property.  Let now $i>1$ and
$t\neq1$. We decompose $t$ as~$t=Xu$, where $X\neq1$ is of minimal
degree such that, for some polynomial $h$, $((i,u),h)\in
G_i$. Let $((i,u_1),g_1),\dots,((i,u_s),g_s)$ be the elements of~$G_i$
distinct from $h$ coming from rows with indices smaller than $(i,t)$
(i.e., $u_j<t$). Following Algorithm~$F_5$, the polynomial $g$ is
obtained from $Xh$ by reductions in the matrix by $h$ and the $g_j$'s,
and $\LT(Xh)\ge \LT(g)$. From line~\ref{ligne:22} in Algorithm
matrix-$F_5$ we see that $\LT(Xh)>\LT(g)$, so by definition of $h$
there exists an index $k$ and a monomial~$\mu$ with
 \[\mu\LT(g_k)=X\LT(h).\]

 By the minimality of~$X$, we have that $X\mid\LT(g_k)$. Since $g_k\in
 G_i$, by Proposition~\ref{prop:LJ} this implies~$X\in\mathcal{T}^i$
 and then also~$\mu\in\mathcal{T}^i$. We prove by contradiction
 that~$x_i\nmid X$.  Again by the minimality of~$X$, if~$x_i\mid X$,
 then~$x_i\nmid\mu$ and therefore~$x_i\mid \LT(g_k)$ and
 $\mu\in\mathcal T^{i-1}$.  By induction there exists a
 monomial~$\tau\in\mathcal{T}^{i-1}$ such that~$g_k$ has
 index~$(i,\tau)$. Now we have that~$((i,Xu),g)$ is reduced by a row
 indexed~$(i,\mu\tau)$ but this is a contradiction
 since~$\mu\tau\in\mathcal{T}^{i-1}$ is such that $\mu\tau\succ Xu$
 for the order grevlex and this would not be a valid row reduction.
\end{proof}

\section{Complexity Analysis of $F_5$}\label{sec2bis}
\subsection{Number of Polynomials}
The following theorem gives quantitative information on the structure
of a reduced Gr\"obner basis (which is independent of the algorithm
used to compute it). To the best of our knowledge, it has not been
given before.
\begin{theorem}\label{fg_gb} Let~$(f_1,\dots,f_m)$ be a homogeneous system for which the variables~$(x_1,\dots,x_n)$ are in simultaneous Noether position, with $d_1=\deg(f_1)\le\dots\le d_m=\deg(f_m)$. Let~$G_i$ be a reduced Gr\"obner basis  of $(f_1,\dots,f_i)$ for the grevlex monomial ordering for $1\le i\le m$.  Then the
number of polynomials of degree $d$ in 
$G_i$ whose leading term does not belong to $\LT(G_{i-1})$
is bounded by~$b_d^{(i)}$, where
 \begin{equation}\label{Bi}
  B_i(z)=\sum_{d = 0}^{\infty}{b_d^{(i)} z^d} = {z^{d_i}\prod_{k = 1}^{i - 1} \frac{1 - z^{d_k}}{1-z}}.
 \end{equation}
For fixed~$i$, let~$D_i=(d_1-1)+\dots+(d_{i-1}-1)$. Then the sequence~$h_d=b^{(i)}_{d+d_i}$, $0\le d\le D_i$ is positive, symmetric ($h_{D_i-d}=h_d$), unimodal ($h_0\le\dots\le h_{\lfloor{D_i/2}\rfloor}\ge\dots\ge h_{D_i}$) and log-concave ($h_d^2\ge h_{d-1}h_{d+1}$).
\end{theorem}
\begin{example}
  In Example~\ref{ex:matrices}, the series $B_3(z) = z^2(1+z)^2 = z^2 +
  2z^3 + z^4$ gives exactly the number of polynomials of degree 2, 3 and 4 in $G_3\setminus G_2$.
\end{example}
\begin{proof}
The proof is by induction on~$i$. If~$i=1$ then by definition of the Noether position, the basis is reduced to one polynomial with leading term~$x_1^{d_1}$ so that in this case~$B_1(z)=z^{d_1}$ as expected. Assuming the property to hold for~$i-1$, we now prove it for~$i$. Consider~$g\in G_i$, which can be written
\begin{equation}\label{decomposition}
g=g_if_i+\dots+g_1f_1,
\end{equation}
for some polynomials~$g_i$. By Proposition~\ref{cor:structure}, we can
restrict our attention to~$g_i$ with~$\LT(g_i)$ belonging
to~$k[x_1,\dots,x_{i-1}]$.

Now if~$g\in G_i$ has degree~$d$ and does not belong to the Gr\"obner basis of~$\langle f_1,\dots,f_{i-1}\rangle$, the number of possible leading monomials of~$g_i$ in such a decomposition is bounded by the number of monomials of degree~$d-d_i$ in~$k[x_1,\dots,x_{i-1}]$ that are not leading terms of a polynomial in~$\langle f_1,\dots,f_{i-1}\rangle$, and this is precisely $\HF_{\theta_{i-1}{\langle f_1,\dots,f_{i-1}\rangle}}(d-d_i)$ with~$\theta_{i-1}$ as in Proposition~\ref{prop:LJ}. This is a bound on the dimension of a vector space containing these elements of~$G_i$. It is therefore also a bound on their possible number of leading terms, whence the result.

The properties of $h_d$ come from the fact that these properties are true for each of the coefficient sequences of the polynomials~$(1-z^{d_k})/(1-z)$ and are preserved by multiplication of these polynomials (see, e.g., \cite{Stanley89}).
\end{proof}

Note that since the series in Theorem~\ref{fg_gb} is actually a polynomial of degree
\[\delta=\sum_{j=1}^i{(d_j-1)}+1,\]
we deduce that this is also a bound on the highest degree of the elements in a reduced grevlex Gr\"obner basis, thus recovering~\cite[Ch.~3, Cor.~3.5]{LJ84}, and from there the result of~\cite{Lazard83} after a generic linear change of variables. 
\begin{corollary}\label{cor:boundD}For a system with variables in simultaneous Noether position, a bound on the number of operations in~$k$ required to compute a Gr\"obner basis for the grevlex ordering is obtained by taking~$D=\sum_{j=1}^m{(d_j-1)}+1$ in Proposition~\ref{propupper}.
\end{corollary}
The next section shows that the~$F_5$ algorithm achieves a smaller complexity.

\subsection{Upper Bound for $F_5$}
We now give a proof of Theorem~\ref{thm:nF5} from page~\pageref{thm:nF5}
using the $F_5$ criterion.
Note that the asymptotic character of this result is only relative
to~$n$: we obtain an actual bound for any fixed degree~$\delta$.

\begin{proof}  
The outline of the proof is as follows. First, we exploit the information on the shape of the Gr\"obner basis from Theorem~\ref{fg_gb} in order to get a good control over the number of operations. Following the structure of the algorithm we get a bound as a sum over~$i=1,\dots,m$, corresponding to the rows induced by each input polynomial~$f_i$, of sums over~$d$, corresponding to the columns induced by the monomials of a given degree. Next, we observe that the (bound on the) cost of all steps over~$i$ is bounded by that devoted to the last polynomial~$f_m$. The situation with respect to degrees is different: there is an intermediate degree where more work takes place. (In fact, we only prove that there is a degree where our bound dominates the other ones, but this behaviour can be observed in practice.) Then, we compute the asymptotic expansion of both this intermediate degree and the (bound on the) number of rows involved in that degree. The final result is obtained by injecting these expansions into the bounds.                                   
\par\smallskip\emph{Exact bound from the numbers of rows and columns in the top-reduction case.}
Arithmetic operations are only performed during the Gaussian
reductions.  If we perform only top-reductions, for each~$1\le i\le m$ and~$\delta\le d \le D$, by
Theorems~\ref{thmf5} and~\ref{fg_gb} there are at most~$b_{d}^{(i)}$
polynomials in~$\mathcal{M}_{d,i}\setminus\mathcal{M}_{d,i-1}$ that
need to be reduced, and they need to be reduced
by~$\tilde{\mathcal{M}}_{d,i-1}$. By Proposition~\ref{cor:structure}
the leading term of the result is in~$\mathcal{T}_d^i$, which implies
that at most~$\binom{i+d-1}{d}$ rows are involved in this reduction,
each row containing~$\binom{n+d-1}{d}$ columns. This gives the following bound on the number of arithmetic operations
\begin{equation}\label{nf5}
N_{F_5}=\sum_{i=1}^m\sum_{d=\delta}^D{b_{d}^{(i)}\binom{i+d-1}{d}\binom{n+d-1}{d}}.
\end{equation}  

\par\smallskip\emph{Bound on the number of polynomials in~$\mathcal{M}_{d,i}\setminus\mathcal{M}_{d,i-1}$.}
For an arbitrary~$d$ and~$i$, we have
\[b_d^{(i)}\le\frac{B_i(r)}{r^d},\qquad \mbox{for any } r>0\] which
follows from the positivity of the coefficients of~$B_i$
(see~\eqref{Bi}). This bound holds for arbitrary $r>0$;  it is
minimised by choosing for~$r$ a root of the derivative of the
right-hand side, which is equivalent (using the logarithmic
derivative) to taking~$r$ such that
 \[d=r\frac{B_i'(r)}{B_i(r)}.\] As
$B_i(r)=r^\delta\left(\frac{1-r^\delta}{1-r}\right)^i$ for equations
of identical degree $\delta$, this is equivalent to 
\begin{equation}\label{rofi}
  i = 1 +
  (d-\delta)\cdot \left(
\frac{\delta }{1-r^{-\delta}}-\frac 1{1-r^{-1}}\right)^{-1}.
\end{equation}
For~$d\ge \delta$, the right-hand side of this equation is a
differentiable function of~$r$, its derivative is negative, and it has
extreme values~$+\infty$ and~$\frac{d-1}{\delta-1}$. This is then a
positive decreasing function of~$r$. This shows that~\eqref{rofi}
defines~$r$ as a function~$r(i,d)$ for~$i\ge \frac{d-1}{\delta-1}$ and
$\delta\le d$, and this function is a positive, decreasing,
differentiable function of~$i$.

The same reasoning on the equation~$d=r\frac{B_i'(r)}{B_i(r)}$ shows
that~$r(i,d)$ is differentiable with respect to $d$, for~$\delta\le
d\le \deg B_i$, and that it is positive and increasing.

\par\smallskip\emph{Bound on the work on the~$i$th polynomial.}
A bound on the sequence summed in~\eqref{nf5} is now obtained by
bounding the \emph{differentiable} function of~$i$
\[\frac{B_i(r(i,d))}{r(i,d)^d}\binom{i+d-1}{d}\binom{n+d-1}{d}\]
for $1\le i \le m$ such that $b_d^{(i)}\neq 0$, i.e., $i\ge
\frac{d-1}{\delta-1}$. This is an increasing function of~$i$. Indeed, 
the last binomial does not depend on~$i$, the previous one is clearly increasing and the logarithmic derivative of the first
factor w.r.t. $i$ is
\begin{equation}
\frac{1-r(i,d)^\delta}{1-r(i,d)} = 1 + r(i,d) + \cdots + r(i,d)^{\delta-1}
\end{equation}
which is positive. Hence, we have the bound
\begin{equation}
  N_{F_5} \le m \sum_{d=\delta}^{D} \frac{B_m(r(m,d))}{r(m,d)^d}\binom{m+d-1}{d}\binom{n+d-1}{d}.
\end{equation}
(Intuitively, the most expensive part
of the computation is performed for $i=m$.)

\par\smallskip\emph{The most expensive degree~$d$.}
Consider now the summand as a function of $d$. Its logarithmic
derivative in $d$ for fixed~$r$ has a simple expression, which vanishes for $d$ root of the equation
\begin{equation}\label{eq_delta}
2\psi(d+m)-2\psi(d+1)=\log r(m,d),
\end{equation}
where $\psi$ is the logarithmic derivative of the $\Gamma$ function (see, e.g., \cite[Ch.~6]{AbSt73}). Thus we now have two equations, \eqref{eq_delta} and~\eqref{rofi} (with $i=m$), relating~$r,d,m$.

As a consequence of the functional equation of the $\Gamma$ function ($\Gamma(s+1)=s\Gamma(s)$), the left-hand side of~\eqref{eq_delta} can be rewritten using
the alternative expression
\[\psi(d+m)-\psi(d+1)=\frac{1}{d+1}+\dots+\frac1{d+m-1}.\]
It follows that Equation~\eqref{eq_delta} defines a unique positive,
differentiable function $d(m)$, such that the pair $(d(m), \rho(m))$
with the notation 
\[\rho(m):=r(m,d(m)))\] gives the solution
to~(\ref{rofi},\ref{eq_delta}) (for $i=m$). Moreover, we have
$\rho(m)\ge 1$. 

We have now isolated the most expensive step of the algorithm and basically bound all of them by it.
The total number of arithmetic operations in $k$ required by algorithm
matrix-$F_5$ is therefore bounded by
\begin{equation}\label{boundN}
  N_{F_5}\le
  m\,(m\,(\delta-1)+1-(\delta-1))\,\frac{B_m(\rho(m))}{\rho(m)^{d(m)}}
  \binom{m+d(m)-1}{d(m)}\binom{m+\ell+d(m)-1}{d(m)},\\
\end{equation}
where we have used the bound on~$D$ from Corollary~\ref{cor:boundD} and the hypothesis $n=m+\ell$.

\par\smallskip\emph{Asymptotic expansions.}
We now let $m\rightarrow\infty$ and obtain the asymptotic behaviour of both~$\rho$ and~$d$ simultaneously, before injecting into the bound above.

Rewriting~\eqref{rofi} yields
\begin{equation}
  \label{dofm}
  \frac{d(m)-\delta}{m-1} = 
\frac{\delta}{1-\rho(m)^{-\delta}} -\frac 1{1-\rho(m)^{-1}},
\end{equation}
from which we define
\begin{equation*}
  \lambda(m):=\frac{d(m)}m.
\end{equation*}
Now, since $\rho(m)\ge1$, the right-hand side of Equation~\eqref{dofm}
takes its values between $\frac{\delta-1}2$ and $\delta-1$, which
implies that $\lambda(m)$ is bounded as $m\to\infty$.

Thus we can compute the asymptotic behaviour of~\eqref{eq_delta}, using the classical expansion
\begin{equation}\label{asympt_psi}
\psi(x) = \log(x) -
\frac{1}{2x} - \sum_{n=1}^\infty \frac{B(2n)}{2n(x^{2n})},\qquad x\rightarrow+\infty
\end{equation}
where $B(n)$ is the $n$th Bernoulli number. This gives
\[\rho(m) =
\left(1+\lambda(m)^{-1}\right)^2
+O\left(\frac1m\right)
\]
We can then eliminate~$\rho(m)$ from~\eqref{dofm} asymptotically:
\[
\lambda(m)+O \left( \frac1{m} \right) =
\frac{\delta}{ 1-\left(1+\lambda(m)^{-1}\right) ^{-2\delta}}
-\frac{1}{1-\left(1+\lambda(m)^{-1}\right)^{-2}}+O \left(
  \frac1{m} \right)
\]
from which we see that as $m\to\infty$, $\lambda(m)$ tends to $\lambda_0$ the unique root
between $\frac{\delta-1}2$ and $\delta-1$ of
\[
\lambda_0 =
\frac{\delta}{ 1-\left(1+\lambda_0^{-1}\right) ^{-2\delta}}
-\frac{1}{1-\left(1+\lambda_0^{-1}\right)^{-2}}
\]
or equivalently
\begin{equation}\label{eq_lambda_0}
\left(1+\lambda_0^{-1}\right)^{2\delta} = \frac1{1-\delta\left(\frac{({\lambda_0}+1)^2-{\lambda_0}^2}{({\lambda_0}+1)^3-{\lambda_0}^3}\right)}.
\end{equation}
Now, Equation~\eqref{eq_delta} using~\eqref{asympt_psi}
gives $\rho$ as a bivariate formal expansion in $1/m$ and $\lambda-\lambda_0$,
with
\[\rho =
\left(1+\lambda_0^{-1}\right)^2
-\frac{2(\lambda_0+1)}{\lambda_0^3}(\lambda-\lambda_0) -
\frac{(\lambda_0+1)\;(2\;\lambda_0+1)}{\lambda_0^3}\frac1m +
\sum_{i,j\ge1} \rho_{i,j}(\lambda_0)\frac{(\lambda-\lambda_0)^i}{m^j}.\]
Injecting into Equation~\eqref{dofm} yields a bivariate formal power series
$\Phi(1/m,\lambda-\lambda_0)=0$ where
\begin{multline*}
  \Phi(1/m,\lambda-\lambda_0) =
\left( {\frac {3\,{{\lambda_0}}^{2}+3\,{\lambda_0}+2}{ \left( {\lambda_0}+1 \right) {\lambda_0}}}-{\frac { 2\,\left( 3\,{{\lambda_0}}^{2}+3\,{\lambda_0}+1 \right)}{{\lambda_0}\, \left( {\lambda_0}+1 \right)  \left( 2\,{\lambda_0}+1 \right) }} \delta\right) (\lambda-\lambda_0) +\\
 \left( {\frac {3\,{\lambda_0}^{3}+5\,{\lambda_0}^{2}+4\,\lambda_0+1}{\lambda_0\, \left( \lambda_0+1 \right) }}-{\frac { \left( 2\,{\lambda_0}^{2}+2\,\lambda_0+1 \right)}{\lambda_0\, \left( \lambda_0+1 \right) }} \delta  \right) \frac{1}{m}+O \left( \frac{\lambda-\lambda_0}{m} \right)
\end{multline*}
The asymptotic implicit function theorem (see, e.g., \cite{GerardJurkat1992}) implies that the asymptotic expansion of $\lambda-\lambda_0$ is given by the formal power series $S\in \mathbb C[[1/m]]$ such
that $\Phi(1/m,S(1/m)) = 0$. This series can be computed to arbitrary order, e.g., by indeterminate coefficients or Newton
iteration, and for instance we get
\begin{align*}
\lambda(m) &= {\lambda_0}-{\frac { \left( 2\,{\lambda_0}+1 \right)
      \left(  \left( 2\,{{\lambda_0}}^{2}+2\,{\lambda_0}+1 \right)
        {\delta}-3\,{{\lambda_0}}^{3}-5\,{{\lambda_0}}^{2}-4\,{\lambda_0}-1
      \right) }{ \left( 6\,{{\lambda_0}}^{2}+6\,{\lambda_0}+2 \right)
      {\delta}-6\,{{\lambda_0}}^{3}-9\,{{\lambda_0}}^{2}-7\,{\lambda_0}-2}}\frac1m+O
  \left( {\frac1{m^2}} \right),\\
  \rho(m) &={\frac { \left( {\lambda_0}+1 \right) ^{2}}{{{\lambda_0}}^{2}}}-{
    \frac { \left( {\lambda_0}+1 \right) ^{2} \left( 2\,{\lambda_0}+1
      \right)  \left( 2\,{\delta}+1 \right) }{{{\lambda_0}}^{2}
      \left(  \left( 6\,{{\lambda_0}}^{2}+6\,{\lambda_0}+2 \right) {\delta}-6\,{{\lambda_0}
        }^{3}-9\,{{\lambda_0}}^{2}-7\,{\lambda_0}-2 \right) }}\frac1m+O \left( {\frac1{m^2}} \right).
\end{align*}

\par\smallskip\emph{Final bound.} We now inject these estimates into~\eqref{boundN} and obtain the final result:
\[
N_{F_5}\le A\,m\,\left( \frac{\left(\frac{{\lambda_0}+1}{{\lambda_0}}\right)^{2\delta}-1}
  {\frac{1}{{\lambda_0}^2}-\frac{1}{({\lambda_0}+1)^2}} 
\right)^m\left(1+O\left(\frac1m\right)\right),
\]
where $A$ is a constant depending only on $\delta$, ${\lambda_0}$ and $\ell$, whose value is given by
\[
A = \frac{1-\delta^{-1}}{2 {\pi }}\cdot
\frac{\left(1+\lambda_0^{-1}\right)^3-1}{(1+\lambda_0)^{1+\ell}}.
\]

\par\smallskip\emph{Inequalities.}
The proof of Theorem~\ref{thm:nF5} is concluded by showing that the expression abbreviated~$B(\delta)$ in the theorem remains within the interval~$[\delta^3,3\delta^3]$.

The lower bound is obtained by observing that, in view of~\eqref{eq_lambda_0},
\begin{align*}
B-\delta^3&=\frac{\delta\lambda_0^2(\lambda_0+1)^2}{1-\delta+3\lambda_0-2\delta\lambda_0+3\lambda_0^2}-\delta^3\\
&=\frac{\delta(\delta-\lambda_0)(\delta-1-\lambda_0)(\lambda_0+\delta+2\delta\lambda_0+\lambda_0^2)}{1-\delta+3\lambda_0-2\delta\lambda_0+3\lambda_0^2}\\
&=\frac{\delta(\delta-\lambda_0)(\delta-1-\lambda_0)(\lambda_0+\delta+2\delta\lambda_0+\lambda_0^2)}
{((1+\lambda_0)^3-\lambda_0^3)(1+\lambda_0^{-1})^{-2\delta}}
\end{align*}
and this is nonnegative since $\lambda_0$ is smaller than $\delta-1$.

The upper bound is obtained in two steps: first $B(\delta)/\delta^3$ is shown to be increasing with~$\delta$ and next its limit as $\delta\rightarrow\infty$ is computed. Equation~\eqref{eq_lambda_0} shows that~$\lambda_0$ is a differentiable function of~$\delta$ and thus so is~$B$. An expression for~$\lambda_0'(\delta)$ is obtained by differentiating~\eqref{eq_lambda_0}. Injecting this expression into the derivative of~$B$ and simplifying gives the logarithmic derivative of $B(\delta)/\delta^3$:
\[\frac{B'(\delta)}{B(\delta)}-\frac3\delta=\frac{(3\lambda_0^{2}+3\lambda_0+1)\ln((1+\lambda_0^{-1})^2)-3(2\lambda_0+1)}{ \left( 2{ 
\lambda_0}+1 \right) {\delta}}.
\]
The numerator depends on~$\lambda_0$ only. Its positivity is obtained by using the lower bound $\ln(1+x)\ge x-x^2/2+x^3/3-x^4/4$ for $x\in(0,1)$, leading to a polynomial that is positive as soon as $\lambda_0>1.4$. Since $\lambda_0\ge(\delta-1)/2$, the monotonicity of~$B(\delta)/\delta^3$ is then established by checking the values for the cases $\delta=2,3,4$ and indeed, we obtain the approximations $2.45$, $2.66$, $2.73$
that conclude this part.

The limit of~$B(\delta)/\delta^3$ is obtained in a way similar to the asymptotic expansions above. We first consider the asymptotic behaviour of~$\lambda_0$ as~$\delta\rightarrow\infty$. Setting~$\lambda_0=c\delta$ and taking the asymptotic expansion of~\eqref{eq_lambda_0} yields
\[e^{2/c}+O(1/\delta)=\frac{3c}{3c-2}+O(1/\delta).\]
The limiting value of~$c$ is therefore the solution of the equation given by the leading term. It can be expressed in terms of the Lambert~$W$ function as
\[\lim_{\delta\rightarrow\infty}\frac{\lambda_0}{\delta}=\frac{2}{3+W(-3e^{-3})}\simeq0.708858.\]
In terms of this value of the limit of~$c$, a direct computation gives
\[\lim_{\delta\rightarrow\infty}\frac{B(\delta)}{\delta^3}=\frac{-4}{W(-3e^{-3})(3+W(-3e^{-3}))^2}\simeq2.81405669.\]
This concludes the proof that~$B(\delta)\le 3\delta^3$, with a more precise estimate in place of the factor~3.
\end{proof}
\noindent Note that full asymptotic expansions of all the parameters can be derived along the same lines.

\section{Practical results}
\label{sec3}
In order to estimate the accuracy or the lack of preciseness of our complexity
result, we performed actual Gröbner bases computations for quadratic and
cubic systems of $n$ equations in $n$ unknowns, in simultaneous
Noether position, with dense coefficients.  These results have been
obtained on a PC\ (laptop) with $4$ GB of RAM.  The
$F_{5}$ algorithm has been implemented in the C language
within the FGb software~\cite{F10c} and we used this implementation to
compute the Gr\"obner basis and to count the exact number of
arithmetic operations (multiplications of integers modulo
$p=65521$). These results show that the estimate of the number of
polynomials in the Gröbner basis (Eq.~\eqref{Bi}, Theorem~\ref{fg_gb})
and the asymptotic bound of $N_{F_5}$ (Eq.~\eqref{eq:asymptNF5},
Theorem~\ref{thm:nF5}) are very precise, but that the quantity
$N_{F_5}$ (Equation~\eqref{nf5}) could be sharpened. This will be done
in a future work.

\par\smallskip\emph{Accuracy of the asymptotic estimate of $N_{F_5}$.}
We plot in Figure~\ref{fig:asympt} the values of $N_{F_5}$ computed
from Eq.~\eqref{nf5} and its asymptotic estimate given in
Theorem~\ref{thm:nF5}, Eq.~\eqref{eq:asymptNF5}. This
shows that the asymptotic bound for $N_{F_5}$ is accurate.
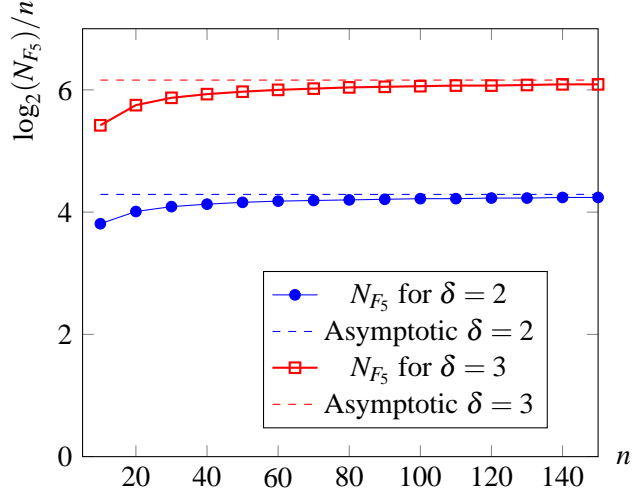
\begin{figure}[h]
\centering
\begin{tikzpicture}
\begin{axis}[
xmin=5,
xmax=150,
ymin=0,
ymax=7,
xlabel style={at={(1.05,0.05)},anchor=south}, 
ylabel style={at={(0.07,1.09)},anchor=east}, 
xlabel={\(n\)},
ylabel={\(\log_2(N_{F_{5}})/n\)},
cycle list={%
{blue,mark=*},
{blue,dashed},
{red,mark=square,thick},
{red,dashed},
{,
mark options={fill=brown!40},
mark=otimes*}},
legend style={
at={(0.35,0.25)},
anchor=west
}]
\addplot coordinates {
(10., 3.81) (20., 4.01) (30., 4.09) (40., 4.13) (50., 4.16) (60., 4.18) (70., 4.19) (80., 4.20) (90., 4.21) (100., 4.22) (110., 4.22) (120., 4.23) (130., 4.23) (140., 4.24) (150., 4.24)};
\addplot coordinates {
(10, 4.29)  (200, 4.29)};
\addplot coordinates {
(10., 5.42) (20., 5.75) (30., 5.87) (40., 5.93) (50., 5.97) (60., 6.00) (70., 6.02) (80., 6.04) (90., 6.05) (100., 6.06) (110., 6.07) (120., 6.07) (130., 6.08) (140., 6.09) (150., 6.09)
};
\addplot coordinates {
(10, 6.16)  (200, 6.16)};
\legend{$N_{F_{5}}\text{ for }\delta=2$,$\text{Asymptotic }\delta=2$,$N_{F_{5}}\text{ for }\delta=3$,$\text{Asymptotic }\delta=3$}
\end{axis}
\end{tikzpicture}
\caption{\label{fig:asympt} Comparison of \(N_{F_{5}}\) with its asymptotic estimate}
\end{figure}

\par\smallskip\emph{Estimation of $b_d^{(i)}$ in Theorem~\ref{fg_gb}, Equation~\eqref{Bi}.}
We compare the bound $b_d^{(i)}$, the number of polynomials computed
by the matrix \(F_5\) algorithm, and the number of polynomials in the
reduced Gröbner basis. It has been shown that \(F_5\) computes
``redundant'' polynomials that do not belong to the reduced Gröbner
basis (see, e.g.,~\cite{EdGaPe2011}). Moreover, \(F_5\) sometimes
reduces less than $b_d^{(i)}$ polynomials. But experimentally, this
difference is small. When all the equations have the same degree
\(d_i=\delta\), the total number of
polynomials 
in our estimation~\eqref{Bi} becomes:
\begin{displaymath}
\text{Polys}_{F_5}=\sum_{i=1}^{m} \sum_{d=\delta}^{D} b_{d}^{(i)}=\sum_{i=1}^{m}  \delta^{i-1}= \frac{\delta^{m}-1}{\delta-1}.
\end{displaymath}
Figure~\ref{fig:exp4} shows that this bound is very close to the actual
number of polynomials computed by the matrix \(F_5\) algorithm.
\begin{figure}[h]
\centering
\begin{tikzpicture}
\begin{axis}[
xmin=5,
xmax=12,
ymin=0,
ymax=12,
xlabel style={at={(1.05,0.05)},anchor=south}, 
ylabel style={at={(0.07,1.09)},anchor=east}, 
xlabel={\(n\)},
ylabel={\(\log_2(\text{Polys}_{F_5})\)},
cycle list={%
{blue,mark=*},
{blue,dashed},
{red,dashed},
{red,mark=square,thick},
{,
mark options={fill=brown!40},
mark=otimes*}},
legend style={
at={(0.35,0.25)},
anchor=west
}]
\addplot coordinates {
(5, 4.90) (6, 5.93) (7, 6.94) (8, 7.91) (9,8.90) (10,9.87) (11,10.8) (12,11.7)};
\addplot coordinates {
(5, 4.95) (6, 5.97) (7, 6.98) (8, 7.99) (9, 9) (10, 10) (11, 11) (12, 12)};
\addplot coordinates {
(5, 6.92) (6, 8.51) (7,10.1) (8,11.7) (9,13.3) (10,14.9)}; 
\addplot coordinates {
(5, 6.88) (6, 8.44) (7, 10.0)};
\node[pin=-65:{$\textcolor{blue}{\delta=2}$}] at (axis cs:10,10.3) {};
\node[pin=145:{$\textcolor{red}{\delta=3}$}] at (axis cs:7,9.5) {};

\legend{\begin{footnotesize}Exact Number of Polynomials\end{footnotesize} ,\footnotesize{Bound} $ (\delta^{m}-1)/(\delta-1)$ }
\end{axis}
\end{tikzpicture}
\caption{\label{fig:exp4}Comparison of the exact/approximated
number of polynomials computed by \(F_{5}\) }
\end{figure}
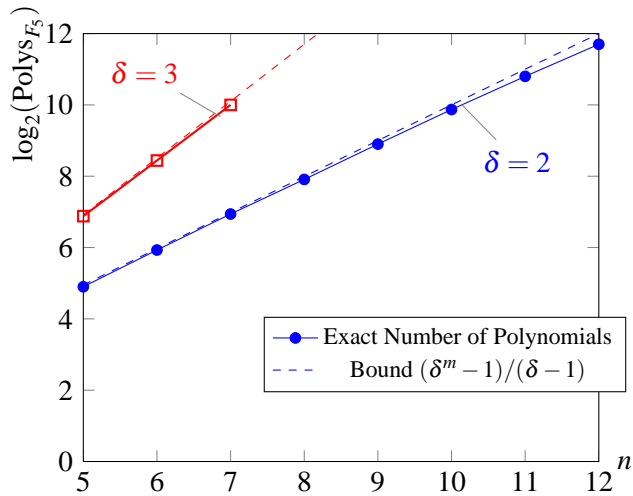

\par\smallskip\emph{Estimation of $N_{F_5}$ in Equation~\eqref{nf5}.}
In Tables~\ref{tab:quad} and~\ref{tab:cub} we report for each \(n\)
the quantities \(\text{top } F_5\) and \(N_{F_5}\), where top \(F_5\)
is the experimental number of multiplications for the matrix version
of \(F_5\) using only top reductions as described in
Section~\ref{sec:desc_F5}, 
and 
\(N_{F_{5}}\) is given by formula \eqref{nf5}. 
These values are represented graphically on Figures~\ref{fig:exp2}
page~\pageref{fig:exp2} and~\ref{fig:exp3} page~\pageref{fig:exp3}.
\begin{table}[h]
   \centering
\begin{tabular}{|l|c|c|c|}
\hline
$n$ & top \(F_5\) &  full \(F_5\) &  \(N_{F_5}\)\\
\hline
7 & \(2^{19.83} =2^{2.83n}\)&\(2^{20.78}=2^{2.97n}\) & \(2^{25.6}=2^{3.65n}\)\\
8 & \(2^{22.95} =2^{2.87n}\)&\(2^{23.14}=2^{ 2.89n}\) & \(2^{29.7}=2^{3.72n}\)\\
9 & \(2^{26.19} =2^{2.91n}\)&\(2^{25.39}=2^{ 2.82n}\)&\(2^{33.9}=2^{3.77n}\)\\
10& \(2^{29.38} =2^{2.94n}\)&\(2^{27.94}=2^{ 2.79n}\)&\(2^{38.1}=2^{3.81n}\)\\
11& \(2^{32.65} =2^{2.97n}\)&\(2^{30.46}=2^{ 2.77n}\)&\(2^{42.3}=2^{3.84n}\)\\
12& \(2^{35.90} =2^{2.99n}\)&\(2^{33.20}=2^{ 2.77n}\)&\(2^{46.4}=2^{3.87n}\)\\
13&                   &\(2^{35.86}=2^{2.76n}\)&\(2^{50.7}=2^{3.9n}\)\\
14&                   &\(2^{38.70}=2^{2.76n}\)&\(2^{54.9}=2^{3.92n}\)\\
15&                   &\(2^{41.52}=2^{2.79n}\)&\(2^{59.1}=2^{3.94n}\)\\
16&                   &\(2^{44.43}=2^{2.78n}\)&\(2^{63.3}=2^{3.96n}\)\\
\hline
\end{tabular}
\caption{\label{tab:quad}Quadratic equations (\(\delta=2\))}
\end{table}

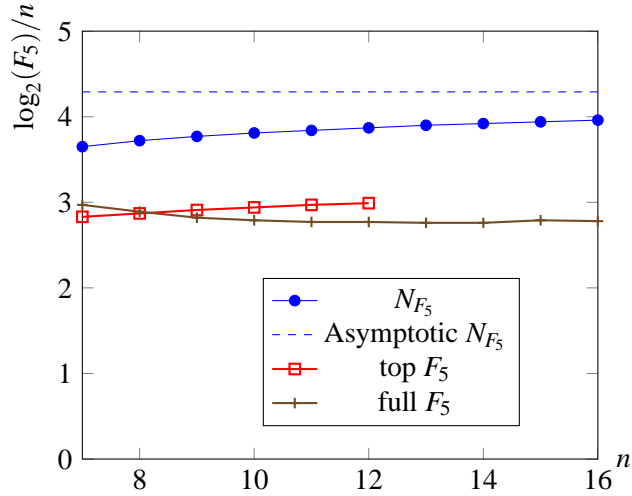
\begin{figure}[h]
\centering
\begin{tikzpicture}
\begin{axis}[
xmin=7,
xmax=16,
ymin=0,
ymax=5,
xlabel style={at={(1.05,0.05)},anchor=south}, 
ylabel style={at={(0.07,1.09)},anchor=east}, 
xlabel={\(n\)},
ylabel={\(\log_2(F_5)/n\)},
cycle list={%
{blue,mark=*},
{blue,dashed},
{red,mark=square,thick},
{brown!60!black,mark=+,thick},
{green!90,thick},
{,
mark options={fill=brown!40},
mark=otimes*}},
legend style={
at={(0.35,0.25)},
anchor=west
}]
\addplot coordinates {
(7., 3.65) (8., 3.72) (9., 3.77) (10., 3.81) (11., 3.84) (12., 3.87) (13., 3.90) (14., 3.92) (15., 3.94) (16., 3.96) (17., 3.97) (18., 3.99) (19., 4.00) (20., 4.01)};
\addplot coordinates {
(7., 4.29)  (20., 4.29)};
\addplot coordinates {
(7., 2.83) (8, 2.87) (9,2.91) (10,2.94) (11,2.97) (12,2.99)
};
\addplot coordinates {
(7., 2.97) (8, 2.89) (9,2.82) (10,2.79) (11,2.77) (12,2.77) (13,2.76)
(14,2.76) (15,2.79) (16,2.78)
};
\legend{$N_{F_{5}}$,$\text{Asymptotic }N_{F_5}$,top
$F_5$,full $F_5$}
\end{axis}
\end{tikzpicture}
\caption{\label{fig:exp2} Experiments with \(\delta=2\)}
\end{figure}

\begin{table}[h]
   \centering
\begin{tabular}{|l|c|c|c|}
\hline
n & top \(F_5\) &  full \(F_5\) & \(N_{F_5}\)\\
\hline
5 & \(2^{20.06}=2^{4.01n}\)&\(2^{22.01}=2^{4.40n}\) & \(2^{24.2}=2^{4.84n}\)\\
6 &\(2^{24.93}=2^{4.16}\)&\(2^{25.85}=2^{4.31n}\) & \(2^{30.1}=2^{5.02n}\)\\
7 &\(2^{29.87}=2^{4.27}\)&\(2^{29.78}=2^{4.25n}\) & \(2^{36.1}=2^{5.16n}\)\\
8 &\(2^{34.86}=2^{4.36}\)&\(2^{33.98}=2^{4.25n}\) & \(2^{42.1}=2^{5.26n}\)\\
9 &\(2^{39.88}=2^{4.43n}\)&\(2^{38.39}=2^{4.27n}\) & \(2^{48.15}=2^{ 5.35n}\)\\ 
10&&\(2^{41.91}=2^{4.29n}\) & \(2^{54.19}=2^{ 5.42n}\)\\ 
\hline
\end{tabular}
\caption{\label{tab:cub}Cubic equations (\(\delta=3\))}
\end{table}

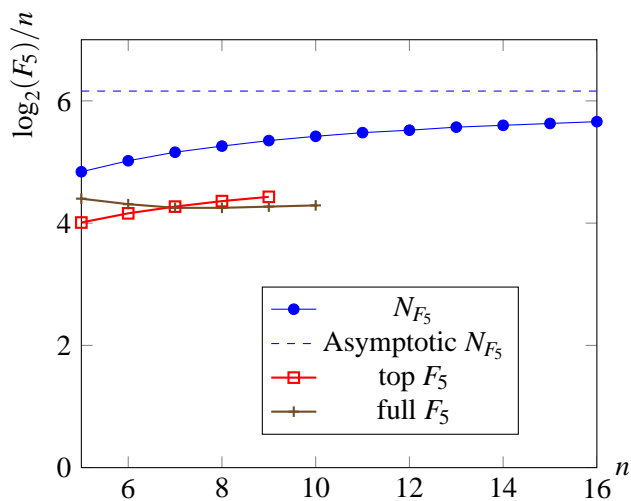
\begin{figure}[h]
\centering
\begin{tikzpicture}
\begin{axis}[
xmin=5,
xmax=16,
ymin=0,
ymax=7,
xlabel style={at={(1.05,0.05)},anchor=south}, 
ylabel style={at={(0.07,1.09)},anchor=east}, 
xlabel={\(n\)},
ylabel={\(\log_2(F_5)/n\)},
cycle list={%
{blue,mark=*},
{blue,dashed},
{red,mark=square,thick},
{brown!60!black,mark=+,thick},
{,
mark options={fill=brown!40},
mark=otimes*}},
legend style={
at={(0.35,0.25)},
anchor=west
}]
\addplot coordinates {
(5., 4.84) (6., 5.02) (7., 5.16) (8., 5.26) (9., 5.35) (10., 5.42) (11., 5.48) (12., 5.52) (13., 5.57) (14., 5.60) (15., 5.63) (16., 5.66) (17., 5.69) (18., 5.71) (19., 5.73) (20., 5.75)};
\addplot coordinates {
(5., 6.16)  (20., 6.16)};
\addplot coordinates {
(5,4.01) (6,4.16) (7., 4.27) (8, 4.36) (9, 4.43) 
};
\addplot coordinates {
(5,4.40) (6,4.31) (7., 4.25) (8, 4.25) (9, 4.27) (10, 4.29)
};

\legend{$N_{F_{5}}$,$\text{Asymptotic }N_{F_5}$,top
$F_5$,full $F_5$}
\end{axis}
\end{tikzpicture}
\caption{\label{fig:exp3} Experiments with \(\delta=3\)}
\end{figure}

This time, we observe a significant gap between the exact values and
our bound. Remember that the bound in Equation~\eqref{nf5}
is computed assuming that $b_d^{(i)}$ polynomials in $\Mac d
i\backslash \Mac d {i-1}$ are reduced (which is exact for dense
polynomials); that they are reduced by
$|\mathcal T_d^i|$ rows, and that each of these rows has $|\mathcal
T_d^n|$ non-zero coefficients.  But in practice each of the
$b_d^{(i)}$ rows needs not be reduced by all of the $|\mathcal T_d^i|$
rows, and some of the $|\mathcal T_d^n|$ entries in each row are
zero. For instance, a row with index $(4,x_3)$ contains no monomial
depending only on $x_1,x_2$, and hence needs not be reduced by the
$|\mathcal T_3^{2}|$ rows whose leading term depend only on $x_1$ and
$x_2$. Another reason is that each row in degree $d$ comes from a row
in degree $d-1$, and has as many non-zero coefficients (before
reduction) as the row in degree $d-1$. Hence the matrices $\Mac d i$
are sparse matrices.
This phenomenon is amplified when the degree $d$ and the number of
variables $n$ grow, and explains the gap between our bound and the
experiments. We did not find  a way to capture this phenomenon in
our bound yet, also because we compute a worst-case bound that has to
capture all the non-generic behaviors.

\par\smallskip\emph{Top reduction vs full reduction.}
Tables~\ref{tab:quad} and~\ref{tab:cub} also contain the quantity full
\(F_5\), which is the exact number of multiplications for Algorithm
\(F'_5\) as described in the original article~\cite{Faugere02}; this
algorithm is a matrix version of \(F_5\) where a full reduction is
used over the rows, instead of only a top reduction: for each row,
reductions are performed not only for the leading term, but also for
all the remaining coefficients.  Experimental results show that this
strategy becomes quickly efficient as $n$ grows, which is due to the
fact, observed empirically, that the rows in the matrices are  sparser with this algorithm
than with  top reduction only.

\section*{Acknowledgement}
 This work was partly supported by the HPAC grant of the French National Research Agency (HPAC ANR-11-BS02-013).

 We thank the anonymous referees for their careful reading and helpful
 comments.
\bibliographystyle{alpha}
\bibliography{paper}

\newcommand{\etalchar}[1]{$^{#1}$}
\begin{thebibliography}{FGLM93}

\bibitem[AS92]{AbSt73}
Milton Abramowitz and Irene~A. Stegun, editors.
\newblock {\em Handbook of mathematical functions with formulas, graphs, and
  mathematical tables}.
\newblock Dover Publications Inc., New York, 1992.
\newblock Reprint of the 1972 edition.

\bibitem[Bar04]{Bardet04}
M.~Bardet.
\newblock {\em {\'E}tude des syst{\`e}mes alg{\'e}briques
  surd{\'e}termin{\'e}s. {A}pplications aux codes correcteurs et {\`a} la
  cryptographie}.
\newblock PhD thesis, Universit{\'e} Paris VI, December 2004.

\bibitem[Buc65]{Buchberger65}
Bruno Buchberger.
\newblock {\em Ein {A}lgorithmus zum {A}uffinden der {B}asiselemente des
  {R}estklassenringes nach einem nulldimensionalen {P}olynomideal}.
\newblock Phd thesis, University of Innsbruck, Austria, 1965.

\bibitem[CLO97]{CoLiOS97}
David Cox, John Little, and Donal O'Shea.
\newblock {\em Ideals, Varieties, and Algorithms}.
\newblock Springer-Verlag, New York, second edition, 1997.

\bibitem[CLO05]{CoxLittleOShea2005}
David~A. Cox, John Little, and Donal O'Shea.
\newblock {\em Using Algebraic Geometry}, volume 185 of {\em Graduate Texts in
  Mathematics}.
\newblock Springer, 2005.

\bibitem[CW90]{CoWi90}
D.~Coppersmith and S.~Winograd.
\newblock Matrix multiplication via arithmetic progressions.
\newblock {\em Journal of Symbolic Computation}, 9(3):251--280, March 1990.

\bibitem[Dub90]{Dube90}
Thomas~W. Dub{\'e}.
\newblock The structure of polynomials ideals and {G}r{\"o}bner bases.
\newblock {\em SIAM Journal on Computing}, 19(4):750--773, August 1990.

\bibitem[EF14]{preprint3}
Christian Eder and Jean-Charles Faug{\`e}re.
\newblock {A survey on signature-based Gr{\"o}bner basis computations}, April
  2014.

\bibitem[EGP11]{EdGaPe2011}
Christian Eder, Justin Gash, and John Perry.
\newblock Modifying {F}aug{\`e}re's {F5} algorithm to ensure termination.
\newblock {\em ACM Commun. Comput. Algebra}, 45:70--89, July 2011.

\bibitem[Fau99]{Faugere99}
Jean-Charles Faug{\`e}re.
\newblock A new efficient algorithm for computing {G}r{\"o}bner bases {$(F\sb
  4)$}.
\newblock {\em Journal of Pure and Applied Algebra}, 139(1-3):61--88, 1999.

\bibitem[Fau02]{Faugere02}
Jean-Charles Faug{\`e}re.
\newblock A new efficient algorithm for computing {G}r{\"o}bner bases without
  reduction to zero ({$F_5$}).
\newblock In Teo Mora, editor, {\em ISSAC 2002}, pages 75--83. ACM Press, 2002.
\newblock Proceedings of the 2002 International Symposium on Symbolic and
  Algebraic Computation, July 07--10, 2002, Universit\'e de Lille, France.

\bibitem[Fau10]{F10c}
J.-C. Faug\`ere.
\newblock {FGb: A Library for Computing Gr\"obner Bases}.
\newblock In Komei Fukuda, Joris Hoeven, Michael Joswig, and Nobuki Takayama,
  editors, {\em {Mathematical Software - ICMS 2010}}, volume 6327 of {\em
  Lecture Notes in Computer Science}, pages 84--87, Berlin, Heidelberg,
  September 2010. Springer Berlin / Heidelberg.

\bibitem[FGLM93]{FaGiLaMo93}
J.~C. Faug{\`e}re, P.~Gianni, D.~Lazard, and T.~Mora.
\newblock Efficient computation of zero-dimensional {G}r{\"o}bner bases by
  change of ordering.
\newblock {\em Journal of Symbolic Computation}, 16(4):329--344, 1993.

\bibitem[FJ03]{FaJo03}
J.-C. Faug{\`e}re and A.~Joux.
\newblock Algebraic cryptanalysis of hidden field equation ({HFE})
  cryptosystems using {G}r{\"o}bner bases.
\newblock In D.~Boneh, editor, {\em Crypto'2003}, number 2729 in Lecture Notes
  in Computer Science, pages 44--60. Springer-Verlag, 2003.

\bibitem[FR09]{FaugereRahmany2009}
Jean-Charles Faug{\`e}re and Sajjad Rahmany.
\newblock {Solving Systems of Polynomial Equations with Symmetries Using
  SAGBI-Gr\"obner Bases}.
\newblock In E.~Kaltofen, editor, {\em ISSAC '09: Proceedings of the 2009
  International Symposium on Symbolic and Algebraic Computation, Seoul Korea,
  ACM}, pages 151--158, 2009.

\bibitem[GH93]{GiHe93}
Marc Giusti and Joos Heintz.
\newblock La d\'etermination des points isol\'es et de la dimension d'une
  vari\'et\'e alg\'ebrique peut se faire en temps polynomial.
\newblock In {\em Computational Algebraic Geometry and Commutative Algebra
  (Cortona, 1991)}, volume XXXIV of {\em Sympos. Math.}, pages 216--256.
  Cambridge Univ. Press, Cambridge, 1993.

\bibitem[GHL{\etalchar{+}}00]{GiHaLeMaSa00}
Marc Giusti, Klemens H{\"a}gele, Gr{\'e}goire Lecerf, Jo{\"e}l Marchand, and
  Bruno Salvy.
\newblock Computing the dimension of a projective variety: the projective
  {N}oether {M}aple package.
\newblock {\em Journal of Symbolic Computation}, 30(3):291--307, September
  2000.

\bibitem[Giu84]{Giusti84}
M.~Giusti.
\newblock Some effectivity problems in polynomial ideal theory.
\newblock In {\em Eurosam 84}, volume 174 of {\em Lecture Notes in Computer
  Science}, pages 159--171, Berlin, 1984. Springer.
\newblock {Cambridge}, 1984.

\bibitem[Giu85]{Giusti85}
Marc Giusti.
\newblock A note on the complexity of constructing standard bases.
\newblock In {\em Eurocal'85}, volume 204 of {\em Lecture Notes in Computer
  Science}, pages 411--412. Springer-Verlag, 1985.

\bibitem[Giu88]{Giusti88}
Marc Giusti.
\newblock Combinatorial dimension theory of algebraic varieties.
\newblock {\em Journal of Symbolic Computation}, 6(2-3):249--265, 1988.
\newblock Special issue on Computational Aspects of Commutative Algebra.

\bibitem[GJ92]{GerardJurkat1992}
R.~G{\'e}rard and W.~B. Jurkat.
\newblock Asymptotic implicit function theorems. {P}art {I}: the preparation
  theorem and division theorem.
\newblock {\em Asymptotic Analysis}, 6:45--71, 1992.

\bibitem[GLS01]{GiLeSa01}
Marc Giusti, Gr{\'e}goire Lecerf, and Bruno Salvy.
\newblock A {G}r{\"o}bner free alternative for polynomial system solving.
\newblock {\em Journal of Complexity}, 17(1):154--211, March 2001.

\bibitem[HL11]{HaLa05}
Amir Hashemi and Daniel Lazard.
\newblock Sharper complexity bounds for zero-dimensional {G}r{\"o}bner bases
  and polynomial system solving.
\newblock {\em International Journal of Algebra and Computation},
  21(05):703--713, 2011.

\bibitem[Hu{\`y}86]{Huynh86}
D{\~u}ng~T. Hu{\`y}nh.
\newblock A superexponential lower bound for {G}r{\"o}bner bases and
  {C}hurch-{R}osser commutative {T}hue systems.
\newblock {\em Information and Control}, 68(1-3):196--206, 1986.

\bibitem[KP96]{KrPa96}
T.~Krick and L.~M. Pardo.
\newblock A computational method for {D}iophantine approximation.
\newblock In {\em Algorithms in Algebraic Geometry and Applications (Santander,
  1994)}, volume 143 of {\em Progr. Math.}, pages 193--253. Birkh{\"a}user,
  Basel, 1996.

\bibitem[Lak91]{Lakshman91}
Y.~N. Lakshman.
\newblock A single exponential bound on the complexity of computing
  {G}r{\"o}bner bases of zero-dimensional ideals.
\newblock In {\em Effective Methods in Algebraic Geometry (Castiglioncello,
  1990)}, volume~94 of {\em Progr. Math.}, pages 227--234. Birkh{\"a}user
  Boston, Boston, MA, 1991.

\bibitem[Laz83]{Lazard83}
D.~Lazard.
\newblock Gr{\"o}bner bases, {G}aussian elimination and resolution of systems
  of algebraic equations.
\newblock In {\em Computer algebra}, volume 162 of {\em Lecture Notes in
  Computer Science}, pages 146--156, Berlin, 1983. Springer.
\newblock Proceedings Eurocal'83, London, 1983.

\bibitem[LG14]{Le-Gall2014}
Fran{\c c}ois Le~Gall.
\newblock Powers of tensors and fast matrix multiplication.
\newblock In {\em ISSAC '14}, 2014.

\bibitem[LJ84]{LJ84}
Monique Lejeune-Jalabert.
\newblock {\em Effectivit{\'e} de calculs polynomiaux}.
\newblock Universit{\'e} de {G}renoble~{I}, 1984.
\newblock Cours de DEA 84--85.

\bibitem[LL91]{LaLa91}
Y.~N. Lakshman and Daniel Lazard.
\newblock On the complexity of zero-dimensional algebraic systems.
\newblock In {\em Effective Methods in Algebraic Geometry (Castiglioncello,
  1990)}, volume~94 of {\em Progr. Math.}, pages 217--225. Birkh{\"a}user
  Boston, Boston, MA, 1991.

\bibitem[Mac02]{Macaulay02}
F.~S. Macaulay.
\newblock On some formul{\ae} in elimination.
\newblock {\em Proceedings of the London Mathematical Society}, 33(1):3--27,
  1902.

\bibitem[Mac16]{Macaulay16}
F.~S. Macaulay.
\newblock {\em The algebraic theory of modular systems}.
\newblock Cambridge Mathematical Library. Cambridge University Press,
  Cambridge, 1916.
\newblock Revised reprint edition in 1994.

\bibitem[MM82]{MaMe82}
Ernst~W. Mayr and Albert~R. Meyer.
\newblock The complexity of the word problems for commutative semigroups and
  polynomial ideals.
\newblock {\em Advances in Mathematics}, 46(3):305--329, 1982.

\bibitem[MM84]{MoMo84}
H.~Michael M{\"o}ller and Ferdinando Mora.
\newblock Upper and lower bounds for the degree of {G}roebner bases.
\newblock In {\em Eurosam 84}, volume 174 of {\em Lecture Notes in Computer
  Science}, pages 172--183, Berlin, 1984. Springer.

\bibitem[Sta89]{Stanley89}
Richard~P. Stanley.
\newblock Log-concave and unimodal sequences in algebra, combinatorics, and
  geometry.
\newblock In {\em Graph theory and its applications: East and West (Jinan,
  1986)}, volume 576 of {\em Annals of the New York Academy of Sciences}, pages
  500--535. New York Academy of Sciences, New York, 1989.

\bibitem[Sto00]{Storjohann00}
Arne Storjohann.
\newblock {\em Algorithms for Matrix Canonical Forms}.
\newblock Phd thesis, Department of Computer Science, ETH, Z{\"u}rich, 2000.

\bibitem[Sto10]{Stothers2010}
A.~Stothers.
\newblock {\em On the Complexity of Matrix Multiplication}.
\newblock PhD thesis, University of Edinburgh, 2010.

\bibitem[VW12]{Vassilevska-Williams2011}
Virginia Vassilevska~Williams.
\newblock Multiplying matrices faster than {C}oppersmith-{W}inograd.
\newblock In {\em S{TOC}'12---{P}roceedings of the 2012 {ACM} {S}ymposium on
  {T}heory of {C}omputing}, pages 887--898. ACM, New York, 2012.

\bibitem[vzGG03]{GathenGerhard2003}
Joachim von~zur Gathen and J{\"u}rgen Gerhard.
\newblock {\em Modern Computer Algebra}.
\newblock Cambridge University Press, New York, 2nd edition, 2003.

\end{thebibliography}
\end{document}